\documentclass{article}%[smallextended]{svjour3}       % onecolumn (second format)
\usepackage{graphicx,amsmath,amsfonts,amsthm,bbm,eurosym,diagbox}
\usepackage{ulem}
\usepackage[colorinlistoftodos]{todonotes}

\newtheorem{theorem}{Theorem}
\newtheorem{proposition}{Proposition}
\newtheorem{lemma}{Lemma}

\theoremstyle{definition}
\newtheorem{assumption}{Assumption}

\newtheorem{definition}{Definition}

\theoremstyle{remark}
\newtheorem{remark}{Remark}

\begin{document}

\title{Price formation and optimal trading in intraday electricity
  markets\thanks{The authors gratefully acknowledge financial support from the Agence Nationale de Recherche
(project EcoREES ANR-19-CE05-0042) and from the
FIME Research Initiative.}}


\author{Olivier F\'eron \\ Electricit\'e de France \\
              Email: \texttt{olivier-2.feron@edf.fr}  \and Peter Tankov  \\ CREST, ENSAE, Institut Polytechnique de
          Paris, France \\ Email: \texttt{peter.tankov@ensae.fr} (corresponding author)\and Laura Tinsi\\ Electricit\'e de France  \\and CREST, ENSAE, Institut Polytechnique de
          Paris, France\\Email: \texttt{Laura.TINSI@ensae.fr}}


\date{}


\maketitle

\begin{abstract}
We develop a tractable equilibrium model for price formation in intraday electricity markets in the presence of intermittent renewable generation.  Using stochastic control theory, we identify the optimal strategies of agents with market impact and exhibit the Nash equilibrium in closed form for a finite number of agents as well as in the asymptotic framework of mean field games.  Our model reproduces the empirical features of intraday market prices, such as increasing price volatility at the approach of the delivery date and the correlation between price and renewable infeed forecasts, and relates these features with market characteristics like liquidity, number of agents, and imbalance penalty.

\end{abstract}



\section{Introduction}

The electricity markets around the world are undergoing a major transformation driven by the transition towards a carbon-free energy system. The increasing penetration of intermittent renewables puts a stronger emphasis on short-term electricity trading and balancing. The intraday electricity markets are increasingly used by the renewable producers to compensate forecast errors. This improves market liquidity and at the same time creates feedback effects of the renewable generation on the market price, leading to increased price volatility and negative correlations between renewable infeed and prices. These effects have an adverse impact on the revenues of renewable producers. They are already significant in countries with high renewable penetration and will become even more important as new renewable capacity comes online. A better understanding of the impact of intermittent renewable generation on intraday electricity market prices and trading volumes is therefore needed to ensure the long-term economic sustainability of the renewable energy production.   

In this paper, we build an equilibrium model for the intraday electricity market, aiming to understand the price formation and identify the optimal strategies for market participants in the setting where both the strategies of the agents and the demand or generation forecasts may affect market prices. We consider an intraday electricity market, where the participants optimize their revenues based on imperfect forecasts of terminal demand or production. We place ourselves in the standard linear-quadratic setting with quadratic trading costs and linear market impact.
The actions of each agent therefore impact market prices, leading to a stochastic game where players interact through  the market price. We exhibit a closed-form Nash equilibrium for this game, and provide explicit formulas for the market price and the strategies of the agents under two different settings:
\begin{itemize}
    \item the setting of $N$ identical agents, having complete information about the forecasts of the other agents,
    \item the setting of an infinite number of identical small agents (the mean field), where each agent only observes the aggregate forecast as well as its own forecast.
\end{itemize}
We then show by theoretical analysis and through numerical simulations that our model reproduces the stylized features of the market price, which we document empirically. In particular,
\begin{itemize}
    \item the market price becomes more volatile at the approach of the delivery time, a phenomenon known as Samuelson's effect in the empirical literature on futures markets;
    \item the market price exhibits negative correlation with the total renewable infeed forecast, which grows in absolute value at the approach of the delivery time. 
\end{itemize}
Furthermore, our model provides direct quantitative links between market characteristics and market price features, as well as the gain of individual agents. For instance, 
\begin{itemize}
    \item observed price volatility increases for higher imbalance penalties which force the agents to follow the forecasts more closely;
    \item observed price volatility increases for lower instantaneous trading costs, which allow agents to trade more actively;
    \item increased competition (greater number of agents in the market) limits profit opportunities for individual agents and leads to lower price volatility.
\end{itemize}

Correlations between renewable infeed and intraday market prices have been studied empirically by a number of authors. Kiesel and Paraschiv \cite{kiesel2017econometric} perform an econometric analysis of the German intraday market and show that a deeper penetration of renewable energies increases market liquidity and price-infeed correlations. The wind power output forecast errors thus turn out to be of paramount importance in explaining the price differences between the day ahead and intraday prices. Karanfil and Li \cite{karanfil2017role} draw similar conclusions from an empirical study of the Danish market, and exhibit the impact of renewable energies on prices, bid-ask spread and volatility. Gruet, Rowińska and Veraart \cite{rowinska2018multifactor} establish a negative correlation between the wind energy penetration and the day ahead market prices.  Jonsson, Pinson and Madsen  \cite{jonsson2010market} show that in addition to creating a negative correlation between the renewable infeed and spot prices, a deeper penetration of the intermittent energies significantly modifies the distribution of spot prices.    

Optimal strategies in the intraday market for a single wind energy producer have also been the object of studies both in the price-taker and price-maker context. In the price-taker setting, Garnier and Madlener \cite{garnier2015balancing} solve a discrete-time optimal trading problem to arbitrate between immediate and delayed trading when price and production forecasts are uncertain. {In \cite{morales2010short}, Morales, Conejo and Perez-Ruiz consider a multimarket setting to derive an optimal bidding strategy for a wind energy producer in the day ahead and adjustment markets, while minimizing the cost incurred in the balancing market. Discrete decisions are taken for each delivery period, considering a finite number of probable scenarii. This approach has been enhanced by Zugno, Morales, Pinson and Madsen, in  \cite{zugno2013pool}, where the wind energy producer is now price maker in the balancing market. Following the same framework, Delikaraoglou, Papakonstantinou, Ordoudis and Pinson  \cite{delikaraoglou2015price} formulate a problem where the renewable producer is price maker in both the day ahead and balancing markets and assess the relevance of strategic behavior in the context of high renewable penetration and varying flexible capacities.} Still in the price-maker setting, continuous-time approaches have also been developed. %{ Electricity markets models developed with this approach are in the spirit of classical financial linear quadratic approaches with price impact as Gârleanu and Pedersen \cite{garleanu2016dynamic} and use target hedging strategies as in Bank et al. \cite{bank2017hedging}.} 
Aïd, Gruet and Pham \cite{aid2016optimal}, consider the optimal trading rate and power generation of a thermal producer when the residual demand at the terminal date is random.  In the same trend, Tan and Tankov \cite{tan2018optimal} develop an optimal trading model for a wind energy producer. They quantify the evolution of forecast uncertainty at the approach of the delivery time, and exhibit optimal strategies depending on forecast updates.

In our study, the uncertain renewable production is also a source of randomness, and the producers' trading decisions impact the market. Unlike the previous papers on electricity markets, we consider the equilibrium setting with many agents and determine the market price as the result of their interaction.  Explicit results for dynamic equilibria are often difficult to obtain. In particular, Nash equilibria often lead to systems of coupled partial differential equations. However, the linear-quadratic setting and in particular the Almgren-Chriss framework of linear market impact and quadratic trading costs has become a standard toolbox allowing many authors to obtain the explicit form of equilibrium price under different market designs.

In the N-agent setting, Bouchard et al.~\cite{bouchard2018equilibrium}, study the equilibrium returns in a market with mean-variance optimizing investors under quadratic transactions costs. Closer to our model, {Vo{\ss}} \cite{voss2019two} considers a game of two agents in the Almgren-Chriss framework, interacting through the market impact function, where each agent aims to follow a target as in the single-agent model of {Bank et al.} \cite{bank2017hedging}. In  \cite{bank2018liquidity}, {Bank et al.} apply the Almgren-Chriss framework to the study of liquidity dynamics in OTC dealer markets. Evangelista and Thamsten \cite{evangelista2020finite} consider liquidation games in a finite population of agents with information asymmetry. 

The problem of finding the equilibrium may be simplified further by assuming a continuum of agents and using the mean field game approach. Fu et al. \cite{fu2021mean} consider the optimal liquidation mean-field game in the generalized Almgren-Chriss framework and obtain the optimal strategies and the equilibrium price as a solution to a linear forward-backward stochastic differential equation (FBSDE) with a singular terminal condition. Fu and Horst \cite{fu2018mean} extend these results to a leader-follower setting using the theory of mean-field games with a major player and Fu et al. \cite{fu2020portfolio} extend the framework with a self-exciting order flow. Fujii and Takahashi \cite{fujii2020mean}  find an equilibrium price under market clearing conditions under quadratic trading costs. Casgrain and Jaimungal \cite{casgrain2020mean,casgrain2019mean} used the Almgren and Chriss framework in the mean field setting to deal with  heterogeneous sub-populations of agents with distinct filtration and/or different beliefs for each sub-population. Shrivats, Firoozi and Jaimungal \cite{shrivats2020mean} recently applied the theoretical setting developed in \cite{casgrain2020mean} to the case of trading in solar renewable energy certificate markets. Finally, while this paper was under review, a different equilibrium price model for electricity market was proposed in {Aïd et al.} \cite{aid2020equilibrium}, where in particular the Samuelson's effect is explained through heterogeneity among agents. 

Following these authors, our paper is based on the Almgren-Chriss toolbox, which we use to study market impact, trading strategies and equilibrium prices in intraday electricity markets. Electricity markets are very different from traditional stock/futures markets. The most fundamental difference is the \emph{predictability} of prices: since electricity is non storable, shifts in demand and supply forecasts are reflected in the price. Other features, such as the Samuelson effect (growth of volatility at the approach of the delivery date) are specific to futures markets in general. Our paper is an attempt to model these features of electricity markets in an endogenous way, by relating price formation to production and demand forecasts. The classical Almgren-Chriss setting, extended with production and demand forecasts provides a simple and tractable framework to take into account the different features of electricity markets: market impact; low liquidity, which improves at the approach of the delivery date etc. We justify the use of this framework a posteriori by showing that the price obtained with our model reproduces the main observed empirical features of intraday electricity prices. 

The main difference of our framework with the existing research, motivated by the predictability of electricity prices, is the presence of forecast processes, which determine the terminal constraint on the strategy of each agent. We consider deterministic market impact and trading cost parameters, which enables us to determine the trading strategies and the equilibrium price in explicit form, under very general assumptions on the price process and the forecast processes. In particular, unlike the above quoted papers, the fundamental price process is not assumed to be a martingale. The explicit form of the equilibrium price enables us to carry out a theoretical study of various characteristics of electricity markets, such as observed volatility, price-forecast correlation, market impact of forecast adjustments and trading costs. All these quantities are also determined in explicit form. As a final contribution, in the last section of the paper we perform an empirical analysis of intraday electricity markets using order book data, and show empirically that the qualitative features of electricity markets are reproduced by our model.

The paper is structured as follows. Section \ref{prelim} describes the market and introduces our modeling framework. 
In Section \ref{Completeinfo} we place ourselves in a setting with a finite number of agents, where all agents observe the forecasts of the other agents. In Section \ref{incompleteinfo} we consider the mean field game, where agents only observe their individual forecasts and the common information. To make a connection between the $N$-agent setting and the mean-field game setting, we show in Section \ref{CV} that (i) the $N$-player equilibrium converges to the mean-field equilibrium as $N\to \infty$, and (ii) an $\varepsilon$-Nash equilibrium for the $N$-player problem may be constructed from the mean-field equilibrium. In Section \ref{insights} we use the results of Section \ref{Completeinfo} to analyze theoretically the properties of equilibrium price in electricity markets. Finally, in Section \ref{empirical} we perform an empirical analysis of intraday electricity prices and confront it with the theoretical results obtained in the preceding sections.

%\todo[inline]{peut etre qu'il faudrait appuyer plus sur la nouveauté du papier}

\section{Preliminaries on electricity markets}
\label{prelim}
In this paper we consider a short-term electricity market, populated with small agents with identical characteristics. These agents face uncertain demand or supply at some future time, and use the electricity market to manage the associated risk.  While our primary interest is to study the impact of increasing renewable penetration on intraday market prices, the market participants may in principle represent both renewable producers with uncertain generation forecasts and industrial consumers / utilities with uncertain demand. To simplify the language and notation, in the sequel, unless specified otherwise, we will refer to forecasts of all agents as demand forecasts (if the agent is a producer, its demand forecast will therefore be negative). 

In most countries, the short-term electricity markets have the following structure (the specific times correspond to the EPEX Spot/Intraday market):
\begin{itemize}
    \item The day-ahead market is a one-off trading venue, where the agents may make bids until 12PM (noon) on the day preceding the delivery day. At 12:55 the price is fixed using the merit order mechanism and the market clears. A major part of the electricity production is sold on the day-ahead market. 
    \item At 3PM on the day preceding the delivery day, the intraday market opens, allowing continuous trading for each quarter-hour of the delivery day. The intraday market has higher trading cost than the day-ahead market, and is mainly used by market participants to adjust their day-ahead positions following forecast updates. 
    \item 15 minutes before delivery the intraday trading for the given delivery period closes. At this point, negative production  imbalances  must  be  compensated  to  the  market  operator  at  the  ’high imbalance settlement price’, which is higher than the last intraday price, and for positive production imbalances, the producer is compensated at the ’low imbalance settlement price’, which is lower than the intraday market price\footnote{See {\footnotesize{\texttt{www.services-rte.com/en/learn-more-about-our-services/\\becoming-a-balance-responsible-party/Imbalance-settlement-price.html}}} for details}. In addition, high imbalances carry a reputational cost for the market participant. To avoid paying the imbalance penalties, the aggregate position (day-ahead plus intraday) held by the agent at the delivery time must therefore be equal to the realized demand. 
\end{itemize}

To represent this market structure in a simplified way, we consider a fixed delivery period, starting at time $T$, and assume that the day-ahead market allows agents to trade instantaneously, without transaction costs, at time $t=0$, at price denoted by $S_0$. Then, between $t=0$ and $t=T$, the agents may trade in the intraday market, at price $(P_t)_{0\leq t\leq T}$, which contains a market impact component, and subject to transaction costs. Finally, at time $T$, if there is an imbalance, the agents must purchase the missing amount /sell the extra amount of electricity at price $S_T$ without transaction cost, and in addition, pay a penalty depending on the absolute value of the imbalance. In the following section we provide details of the model and compute explicitly the optimal strategies of the agents and the equilibrium intraday market price.

\section{Optimal trading strategies and equilibrium price}
In this section we introduce our model of electricity market and derive explicit expressions for the equilibrium price and optimal equilibrium trading strategies of the agents. We consider both the N-agent setting (Section \ref{Completeinfo}) and the mean-field game setting (Section \ref{incompleteinfo}). Section \ref{CV} clarifies the relationship between the equilibrium strategies and prices in the N-agent market and those of the mean field game limit. 
\subsection{N-player setting}\label{Completeinfo}
In this section we assume that in the market there are $N$ identical agents, and we denote by $\phi^i_t$ the position of $i$-th agent at time $t$. As is common in optimal execution literature, we assume that the position of $i$-th agent is an absolutely continuous process, and we define the \textit{rate of trading} $\dot \phi^i_t$.  We introduce a filtered probability space $(\Omega, \mathcal{F}, \mathbb{F}:=(\mathcal{F}_t)_{t \in [0,T]} ,\mathbb{P})$ to which all processes are adapted, and which models the information available in the market to all the agents. 
%The agents control their positions by choosing the trading rate, denoted by $\dot \phi^i_t$ at time $t$ for $i^{th}$ agent. 
{The position of the $i$-th agent at time $t$  is given by $\phi^i_t = \phi^i_0 + \int_{0}^{t}\Dot{\phi^i_s}ds$ with $\phi^i_0 \in \mathcal F_0$ denoting the position of the agent in the day-ahead market.} The fundamental electricity price process is denoted by  $(S_t)_{t \in [0,T]}$, where $S_0$ corresponds to the day-ahead market price and $S_t$ for $0<t<T$ denotes the intraday market price net of the price impact component. 
The intraday market price with the price impact component is denoted by $(P^N_t)_{t \in [0,T]}$. The strategies of the agents impact the market price $P_t^N$ as follows: {\begin{equation}\label{price0}
P^N_t = S_t + a(\bar\phi^N_t-\bar \phi^N_0), \quad \forall t \in [0,T],
\end{equation}}
where $\bar \phi^N_t = \frac{1}{N}\sum_{i=1}^N \phi^i_t$ is the average position of the agents and  $a$ is a constant. The parameter $N$ describes the size of the market (number of agents), it is therefore natural that the trading strategy of each agent has an effect of order of $1/N$ on the market price. 
The permanent component of the price impact of trades in our model is thus linear, which is the only shape  compatible with the absence of arbitrage, see \cite{huberman2004price,gatheral2010no}. On the other hand, the transient component of market impact is not modelled directly. 
{Literature on market microstructure mostly shows that metaorders have a concave transient impact on prices (see Bershova and  Rakhlin \cite{bershova2013non},  Bacry et al.  \cite{bacry2015market}, Bucci et al. \cite{bucci2020co} and Bouchaud \cite{bouchaud2010price}). However, for the sake of simplicity and in order to obtain an analytical solution for our model, we choose a linear impact function as in the seminal papers by Almgren and Chriss \cite{almgren1999value,almgren2001optimal} and many other more recent papers, including  Aïd et al. \cite{aid2016optimal}} in the context of electricity markets. The transient component of the market impact is taken into account indirectly, via a trading cost penalty.

The agents trading in the market at time $t$ incur an instantaneous cost, 
$$
\dot \phi_t^i P^N_t + \frac{\alpha(t)}{2} \dot\phi^i_t{\left(\dot\phi^i_t + b \dot{\bar \phi}^{N,-i}_t\right)}, \quad \forall t \in [0,T]
$$
for the $i$-th agent { where $\dot{\bar \phi}^{N,-i}_t = \frac{1}{N-1} \sum_{j =1, j \ne i}^N \dot{\phi}^{j}_t$}. Here the first term represents the actual cost of buying the electricity, and the second term represents the cost of trading, where $\alpha(.)$ is a continuous strictly positive function on $[0,T]$ reflecting the variation of market liquidity at the approach of the delivery date. { The term $b \dot{\bar \phi}^{N,-i}_t$ with $b>0$ represents the impact of the crowd trading direction on the cost of trading of a single agent, which accounts for possible synchronization of the agents. }  {The instantaneous cost paid by each agent is thus independent of the size of the market. This corresponds to a market where immediately available liquidity (market depth) is low (thus even a minor agent has to pay order book costs) but the order book is resilient (thus the trade of a minor agent  only has a lasting impact of order of $1/N$ on the price). This is consistent with recent empirical and theoretical studies of order book dynamics, for example, according to \cite{donier2015fully}, while the total daily volume
exchange on a typical stock is around 1/200th of its market capitalization, the volume present in the order book at
any given time is 1000 times smaller than this. } 

{Each agent $i$ has a demand forecast $X^i_t$ and aims to maximize her gain from trading in the market under the volume constraint $\phi^i_T = X^i_T$. More precisely, whenever $\phi^i_T \neq X^i_T$, the agent must first purchase the missing amount or sell the extra amount of electricity at price $S_T$ and in addition pay an imbalance penalty $\frac{\lambda}{2}(\phi^i_T - X^i_T)^2$. The actual imbalance mechanism of electricity markets boils down to applying a $L^1$ penalty function to the terminal imbalance; however, large imbalances may also create a reputational damage to the producer, thus a quadratic penalty, which penalizes large imbalances more strongly, appears appropriate. On the other hand, the 'hard constraint' may be recovered from our results by making the penalty parameter $\lambda$ tend to infinity.}

%\hl{These agents are assumed to have taken a position in the day ahead market, and use the intraday market to manage the volume risk associated to the imperfect production forecast used for their day ahead market bid.} 
%They all have the same characteristics. 
%\sout{Each agent has access to her individual production forecast, which is observed by other agents. We assume that the forecast process of $i^{th}$ renewable producer at time $t$ is given by $X^i_t$. }

 %In the following assumption and below, we say that the process $\xi$ is square integrable if $\mathbb E\left[\int_0^T \xi_t^2 dt\right]<\infty$. \todo{Square integrability of $S_0$? }
Our main results hold true under the following assumption. 
\begin{assumption}\label{martass}
The process $S$  is $\mathbb F$-adapted and satisfies
\begin{align}
\mathbb E[\sup_{0\leq t\leq T}S_t^2]<\infty. \label{intS}
\end{align}
and the processes $(X^i)_{i=1}^N$ are square integrable $\mathbb F$-martingales. 
\end{assumption}

Considering the demand forecast as a martingale is natural since it is the best estimate at time $t$ of what the demand will be at the delivery time $T$ given our current knowledge $\mathcal{F}_t$.

\begin{definition}[Admissible strategy]\label{adm.def}We say that the strategy $(\phi^i_t)_{t \in [0,T]}$ of the $i$-th agent is admissible if $\phi^i_0\in \mathcal F_0$, the process $(\dot\phi_t)_{t \in [0,T]}$ is $\mathbb F$-adapted and
$$
\mathbb E\Big[(\phi^i_0)^2 + \int_0^T (\dot\phi^i_t)^2dt\Big]<\infty.
$$
\end{definition}

Following the discussion above, the objective function maximized by agent $i$ is written as follows: 
\begin{align}
       %&J^{N,i}(\phi^i, \phi^{-i}) := - %\mathbb{E}\left[\int_{0}^{T}\left\{\frac{\alpha(t)}{2}(\Dot{\phi^i_t})^{2}+\Dot{\phi^i_{t}}(S_t +a\bar \phi^N_t) %\right\}dt+\frac{\lambda}{2}(\phi^i_{T}- X^i_{T})^2\right]\label{OF},        
       J^{N,i}(\phi^i, \phi^{-i}) := - \mathbb{E}\Big[&\underbrace{\phi^i_0 S_0}_{\text{Day ahead}} + \underbrace{\int_{0}^{T}\left\{\frac{\alpha(t)}{2}\Dot{\phi^i_t}{\left(\dot\phi^i_t + b \dot{\bar \phi}^{N,-i}_t\right)}+\Dot{\phi^i_{t}}P^N_t \right\}dt}_{\text{Intraday}} \notag\\ & \underbrace{- (\phi^i_T-X^i_T)S_T+\frac{\lambda}{2}(\phi^i_{T}- X^i_{T})^2}_{\text{Balancing}}\Big]\label{OF},
\end{align}
where $\phi^{-i}:= (\phi^1,\dots, \phi^{i-1}, \phi^{i+1}, \dots, \phi^N)$  is the vector of positions of all agents except the $i$-th one.  Here, the first term corresponds to the day-ahead market transaction, the integral term corresponds to the cost of purchasing electricity in the intraday market, and the term in the second line corresponds to the imbalance payment.

Because of the price impact, each agent's gain is affected by the decisions of others and we thus face a non-cooperative game. The optimal strategy of each player depends on the other players' actions and we want to describe the resulting dynamical equilibrium, which we define formally below.  

\begin{definition}[Nash Equilibrium]

We say that $({\phi^{i*}_t})_{t \in [0,T]}^{i = 1 \dots N}$ is a Nash Equilibrium for the N-player game if it is a vector of admissible strategies,  and for each $i=1,\dots, N$, 
\begin{align}J^{N,i}( \phi^i,\phi^{-i*}) \leq J^{N,i}( \phi^{i*},
\phi^{-i*})\label{nash}\end{align}
for any other admissible strategy $ \phi^{i}$. 
\end{definition}

%In other words, in the situation of Nash equilibrium, the strategy $\phi^{i*}$ used by each agent is this agent's best response to the strategies $\phi^{-i*}$ of all other agents. 

The following theorem characterizes explicitly the Nash equilibrium of the $N$-player game. In the theorem and its proof, we denote the average forecast process by $\overline X^N_t:= \frac{1}{N}\sum_{i=1}^N X^i_t$ and use the following shorthand notation.
\begin{align}
  \Delta^{N}_{s,t}&:= \int_s^t \frac{\eta^{N}_{u,t}}{\alpha(u){\left(1+\frac{b}{2}\right)}} du\quad &&\text{with}\quad \eta^{N}_{s,t} = e^{-\int_s^t \frac{(N-1)a}{N\alpha(u){\left(1+\frac{b}{2}\right)}}du}\notag\\ 
  \widetilde \Delta^{N}_{s,t} &:=\int_s^t \frac{\tilde \eta^{N}_{u,t}}{\alpha(u)} du,\quad &&\text{with}\quad \tilde \eta^{N}_{s,t} = e^{\int_s^t \frac{a}{N\alpha(u)}du}\notag\\
   I^N_t&:= \int_0^t \frac{\eta^{N}_{s,t}}{\alpha(s) {\left(1+\frac{b}{2}\right)}} S_s ds, \quad &&\widetilde I^N_t := \mathbb E\left[\int_0^T\frac{\eta^{N}_{s,T}}{\alpha(s) {\left(1+\frac{b}{2}\right)}} S_s ds\Big|\mathcal F_t\right],\notag\\
   \widetilde S_t &:= \mathbb E[S_T|\mathcal F_t],\quad &&\check X^i_t = X^i_t - \overline X^N_t.\label{deltanot}
\end{align}

The proof of this theorem can be found in Section \ref{prooft1}.
\begin{theorem}\label{theorem_complete}
Under Assumption \ref{martass}, the unique Nash equilibrium in the $N$-player game is given by
%\begin{align*}
\begin{equation} \label{positionNash}
    \begin{split}
\phi^{i*}_t &=   X^i_0  + \frac{1+ \frac{a}{N}\Delta^{N}_{0,t}}{1+\frac{a}{N} \Delta^{N}_{0,T}} (\widetilde I^N_0-\widetilde S_0 \Delta^N_{0,T})- (I^N_t-\Delta^N_{0,t} \widetilde S_0)
    \\ &+ \int_0^t \Delta^{N}_{s,t}\frac{ \left(\frac{a}{N}+\lambda\right)  d\widetilde I^N_s+ \lambda d\overline X^N_s+d\widetilde S_s}{1+\left(\frac{a}{N}+\lambda\right)  \Delta^{N}_{s,T}}
 + \int_0^t \widetilde \Delta^{N}_{s,t} \frac{\lambda d\check X^i_s}{1+
  \left(\frac{a}{N}+\lambda\right)\widetilde\Delta^{N}_{s,T}}  .
%\end{align*}
    \end{split}
\end{equation}

The equilibrium price has the following form:
\begin{align}
P^N_t &= S_t  + a\frac{\frac{a}{N}\Delta^{N}_{0,t}}{1+\frac{a}{N} \Delta^{N}_{0,T}} (\widetilde I^N_0-\widetilde S_0 \Delta^N_{0,T})\notag- a(I^N_t-\Delta^N_{0,t} \widetilde S_0)
   \\& + a\int_0^t \Delta^{N}_{s,t}\frac{ \left(\frac{a}{N}+\lambda\right)  d\widetilde I^N_s+ \lambda d\overline X^N_s+d\widetilde S_s}{1+\left(\frac{a}{N}+\lambda\right)  \Delta^{N}_{s,T}}.
\label{eq:price_complete_information}
\end{align}
\end{theorem}

\paragraph{Discussion}
The day-ahead market position of $i$-th agent is given by
$$
\phi^{i*}_0 =   X^i_0  + \frac{\widetilde I^N_0-\widetilde S_0 \Delta^N_{0,T}}{1+\frac{a}{N} \Delta^{N}_{0,T}}.
$$
The agents, therefore, trade in the day-ahead market based on their forecasts at time $0$ and apply a correction for the potential fundamental price trend, which disappears if the fundamental price is a martingale. 

For nonzero trading costs, the strategies of the agents and thus the price impact have a finite variation. Hence,  the price impact component does not directly induce additional volatility which may be a weakness of the model. However, the drift $\dot{\Bar{\phi}}^{N}$ is stochastic and thus creates additional price variations, making the effective observed volatility larger. We will investigate this phenomenon in more details  in {Section \ref{marketeffect}}. 

The aggregate intraday market strategy $\bar\phi^{N*}$ (given by equation \eqref{avstrat}) and, consequently, the equilibrium price have a complex structure because of the generality of our setting; in particular the fundamental price process $(S_t)_{0\leq t\leq T}$ is only assumed to be square integrable. Under more stringent assumptions, important simplifications can be obtained, as the following examples illustrate. 
\begin{itemize}
    \item Assume that the fundamental price process $S$ is a martingale. Then, $\widetilde I^N_t  = -\int_0^t S_s d\Delta^N_{s,T}+S_t \Delta^N_{t,T}$ and $d\tilde I^N_t = \Delta^N_{t,T} dS_t$. Substituting this into \eqref{eq:price_complete_information}, after cancellations, we find that aggregate strategy does not depend on the fundamental price:
    $$
    \bar \phi^{N*}_t = \overline X^N_0 +\int_0^t \frac{\Delta^{N}_{s,t}\lambda d\overline X^N_s}{1+\left(\frac{a}{N}+\lambda\right)  \Delta^{N}_{s,T}}.
    $$
    In the absence of price trend, the trades are therefore only provoked by forecast adjustments. 
    \item Assume now that the fundamental price contains a martingale component $M$ and a deterministic component $A$: $S_t = A_t + M_t$. The above argument shows that the aggregate strategy does not depend on the martingale part $M$. Thus, we can assume that $S_t$ is deterministic, which means that $\widetilde I^N_t$ and $\widetilde S_t$ are constant, and the aggregate strategy becomes
    \begin{align*}
    \bar\phi^{N*}_t &=   \overline X^N_0  + \frac{1+ \frac{a}{N}\Delta^{N}_{0,t}}{1+\frac{a}{N} \Delta^{N}_{0,T}}  \int_0^T \frac{\eta^{N}_{s,T}}{\alpha(s) {\left(1+\frac{b}{2}\right)}} (A_s-A_T) ds \\ &- \int_0^t \frac{\eta^{N}_{s,t}}{\alpha(s) {\left(1+\frac{b}{2}\right)}} (A_s-A_T) ds
    + \int_0^t \frac{  \Delta^{N}_{s,t}\lambda d\overline X^N_s}{1+\left(\frac{a}{N}+\lambda\right)  \Delta^{N}_{s,T}}
    \end{align*}
    If there is a positive trend in the fundamental price (that is, $A$ is increasing), then the day-ahead position will be below the demand forecast, but there will be a positive trend in the aggregate strategy: the overall price trend will be amplified by the market impact component. 
    \item Consider now the limiting case of infinite penalty: $\lambda\to \infty$. Then, using the dominated convergence as needed, we see that the aggregate strategy satisfies:
    \begin{align*}
    \lim_{\lambda \to \infty}\bar\phi^{N*}_t &=   \overline X^N_0  + \frac{ 1+\frac{a}{N}\Delta^{N}_{0,t}}{1+\frac{a}{N} \Delta^{N}_{0,T}} (\widetilde I^N_0-\widetilde S_0 \Delta^N_{0,T})\\&- (I^N_t-\Delta^N_{0,t} \widetilde S_0)
    + \int_0^t \Delta^{N}_{s,t}\frac{d\widetilde I^N_s+  d\overline X^N_s}{ \Delta^{N}_{s,T}}
    \end{align*}
\item  Finally, let us compute the form of trading strategy in the  limit of zero trading costs. To this end, we assume in addition that the fundamental price process $S$ has a left limit at every point.    Fixing $s<t \in [0,T]$, we have:
$$ \Delta^{N}_{s,t}= { \frac{N}{a(N-1)}\left(1 - e^{-\int_s^t\frac{a(N-1)}{\alpha(l){\left(1+\frac{b}{2}\right)}N}dl}\right)\longrightarrow \frac{N}{a(N-1)}}:=\Delta^*$$
as $\alpha(t)\to 0$ uniformly in $t$.  From the left limit property of $S$, it is easy to see that $I^N_t \to \Delta^* S_t$ almost surely, for every $t$. For similar reasons, using the dominated convergence theorem, $\widetilde I^N_t \to \Delta^*  \widetilde S_t$. Finally,
%\begin{align*}
%\lim_{\|\alpha\|\to 0} \bar \phi^N_t &= -\Delta^* S_t + \frac{\Delta^*}{ 1+\left(\frac{a}{N}+\lambda\right)\Delta^*} \left(\lambda\overline X^N_t + \left(\frac{a}{N}+\lambda\right)\Delta^* \mathbb E[S_T|\mathcal F_t]\right)\\
%& = \frac{\lambda}{a+\lambda}\overline X^N_t -\frac{N}{a(N-1)} S_t + \frac{\frac{a}{N}+\lambda}{a+\lambda}  \frac{N}{a(N-1)} \mathbb E[S_T|\mathcal F_t]
%\end{align*}

\begin{align*}
    \lim_{\|\alpha\|\to 0}\bar\phi^{N*}_t =   \overline X^N_0  
    + \frac{N}{a(N-1)}(\widetilde S_t - S_t) + \frac{\lambda}{a+\lambda}(\overline X^N_t - \overline X^N_0)
\end{align*}
Thus, in the absence of trading costs, for $N\geq 2$, the aggregate equilibrium strategy is well defined, and the gain of each agent remains bounded in expectation. This is in contrast with the single-agent case, where the gain may be arbitrarily large, unless the fundamental price process is a martingale. Indeed, in the single-agent case, without transaction costs the objective function writes:
\begin{align*}
J^{1,i}(\phi) =  \mathbb E\left[\int_0^T \phi_t d S_t  - \frac{a}{2}(\phi_T-\phi_0)^2- \frac{\lambda}{2}(X_T - \phi_T)^2 - X_T S_T\right],
\end{align*}
and it is clear that unless the fundamental price process is a martingale, this expression can be made arbitrarily large. This means that the "price of anarchy" in this model is infinite: if the agents chose the same strategy, they could have all obtained an infinite gain, but competition between agents limits everybody's gain to a finite value. 

\item Finally, coming back to the form of the individual agent's strategy $\phi^{i*}_t$, we see that the dependence on the trading cost $\alpha$ is different for the common part of the strategy  and the individual part of the strategy (the last term of the formula). While the common part of the strategy depends on the "effective trading cost" $\alpha(1+b/2)$, taking into account the crowd behavior, the individual part of the strategy depends only on $\alpha$. We conclude that due to additional costs related to crowd behavior of agents, the agents trade less actively in response to common forecast updates than in response to individual forecast updates. 
\end{itemize}

{

\subsection{Mean-field game setting}\label{incompleteinfo}
%In this section we consider that the information available to agents is no longer the same. 
In this section, we place ourselves in the mean field game limit, that is, we assume the number of agents in the market, $N$ tends to infinity, while the strategy of each agent remains finite. We then consider a generic agent
%We consider now a continuum of agents with identical characteristics.  
%The information available to agents is no longer the same. We consider a generic agent, 
and denote by $X :=(X_t)_{t\in [0,T]}$ the demand forecast of this agent, by $\phi$ the agent's position and by $\mathbb F$ the filtration which contains the information available to this agent. In addition we introduce a smaller filtration, containing the common noise and denoted by $\mathbb F^0$. This filtration contains the information about the fundamental price and potentially some information about the demand forecast but, in general, not the full individual demand forecast of the generic agent. {We decompose the individual demand forecast as follows: $X_t = \overline{X}_t + \check{X}_t$, where $\overline{X}_t = \mathbb{E}\left[X_t |\mathcal{F}_t^0 \right]$ is common for all agents (it can be seen as a national demand forecast).} In this mean field game setting, 
% In this section, we consider the mean field game limit, that is, we assume that the number of agents in the market, $N$ tends to infinity, while the strategy of each agent remains finite. The agents are assumed to be identical and independent, conditionally to the common noise. We directly place ourselves in the mean field game limit, and therefore, assume that 
the average quantities of the $N$-agent problem are replaced with conditional expectations with respect to the common noise filtration $\mathbb F^0$. 

For any $\mathbb{F}$-adapted process $(\zeta_t)_{t \in [0,T]}$, we will denote $\bar \zeta_t = \mathbb{E}[\zeta_t |\mathcal{F}^0_t] = \int_{\mathbb{R}}x \mu_t^{\zeta}(dx)$ where:   $\mu^{\zeta}_t:= \mathcal{L}(\zeta_t|\mathcal{F}^0_t)$.
The game is now represented by the interaction of agents through the conditional distribution flow $\mu^{\phi}_t:= \mathcal{L}(\phi_t|\mathcal{F}^0_t)$ of the state process. % given the common noise information. 
% We have the following analogy between the $N$-player game and the mean field one:
%$$\frac{1}{N} \sum_{i=1}^N \phi^i_t = \int_\mathbb{R}x \delta_{\phi^i}(dx),
%\leftrightarrow \mathbb{E}[\phi_t|\mathcal{F}^0_t] = \int_\mathbb{R} x \mu^\phi_t (dx) $$ where $\delta_x(.)$ is the Dirac measure at the point $x$. 
%\sout{In the price impact function defined in the previous section, expectation with respect to the empirical measure, will be replaced by an integral with respect to the measure flow such that the market price is now given by:}
{The price impact function, defined in the previous section as an expectation with respect to the empirical measure, is now an integral with respect to the measure flow:}
\begin{equation}
    P_t = S_t + a (\bar \phi_t - \bar \phi_0).
\end{equation}
Each individual agent now has a negligible impact on the price, but the aggregate position of all agents has a nonzero impact. Thus, in the mean-field game setting, we consider that the market is very large compared to the size of the individual agent, but the immediately available liquidity in the order book is small, so that even a minor agent pays a non-zero trading cost.

 The objective function for the generic agent is
\begin{align} \label{mfproblem}
    J^{MF}(\phi, \bar \phi) :=   - \mathbb{E}\Big[&\phi_0S_0 + \int_0^T  \frac{\alpha(t)}{2}\dot \phi_t{\left(\dot \phi_t+b \dot{\bar \phi}_t\right)} + \dot\phi_t (S_t + a(\bar \phi_t-\bar \phi_0))dt \notag\\ -&(\phi_T-X_T) S_T+ \frac{\lambda}{2}(\phi_T-X_T)^2\Big].
\end{align}
As in the previous section, each agent maximizes this functional over the set of strategies satisfying Definition \ref{adm.def}. 

%Similarly to the previous section, we shall say that the strategy of the generic agent $(\dot\phi_t)_{0\leq t\leq T}$ is admissible if it is $\mathbb F$-adapted and square integrable. 

We now define the mean field equilibrium.
\begin{definition}[mean field equilibrium]\label{defMFE} 
An admissible strategy\\  $\phi^* := (\phi^*_t)_{t \in [0,T]}$ is a mean field equilibrium if for any admissible strategy $\phi$, 
$$J^{MF}(\phi,\bar \phi^*)\leq J^{MF}(\phi^*,\bar \phi^*).$$ 
\end{definition}

In this section, we make the following assumption. 
\begin{assumption}\label{mfgass}${}$
\begin{itemize}
\item The process $S$ is adapted to the filtration $\mathbb F^0$ and satisfies \eqref{intS}.
\item The process $X$ is a square integrable martingale with respect to the filtration $\mathbb F$. 
\item The process $\overline X$ defined by $\overline X_t := \mathbb E[X_t|\mathcal F^ 0_t]$ for $0\leq t\leq T$ is a square integrable martingale with respect to the filtration $\mathbb F$.
\end{itemize}
\end{assumption}
Note that if $X$ is an $\mathbb F$-martingale, then $\overline X$ is by construction an $\mathbb F^0$-martingale, but it may not necessarily be a martingale in the larger filtration $\mathbb F$.

The following theorem characterizes the mean field equilibrium in our setting. The statement of the theorem appears similar to that of Theorem \ref{theorem_complete}, modulo replacing $\overline X^N$ with $\overline X$ and making $N$ tend to infinity. %Thus the main implications of Theorem \ref{theorem_complete}, given in Corollary \ref{limit.cor}  and section \ref{marketeffect}, hold true also in this case, except, of course for the ones which describe the behavior of the market price as the number of agents tends to infinity. 
However, the computation of the strategy and the market price in the N-player setting requires the knowledge of the sum of forecasts of all agents whereas in the mean-field setting one needs to know the conditional expectation of the agent's forecast with respect to the 'common knowledge' filtration. 
Thus, the theoretical price given by this theorem can be computed by the regulator, and the strategy of this theorem can be computed by an individual player, both of which do not have the complete information about the forecasts of other players.

%In the theorem and its proof, we use the same shorthand notation \eqref{deltanot} as before.
{In the theorem and its proof, we use the following shorthand notation.
\begin{align}
  &\Delta_{s,t}:= \int_s^t \frac{\eta_{u,t}}{\alpha(u){\left(1+\frac{b}{2}\right)}} du\quad &&\text{with}\quad \eta_{s,t} = e^{-\int_s^t \frac{a}{\alpha(u) {\left(1+\frac{b}{2}\right)}}du}, \notag\\
   &I_t:= \int_0^t \frac{\eta_{s,t}}{\alpha(s){\left(1+\frac{b}{2}\right)} } S_s ds, \quad &&\widetilde I_t := \mathbb E\left[\int_0^T\frac{\eta_{s,T}}{\alpha(s) {\left(1+\frac{b}{2}\right)}} S_s ds\Big|\mathcal F_t\right], \notag\\
   &\text{and}  \quad \widetilde \Delta_{s,t} :=\int_s^t \alpha^{-1}(u) du.
\end{align}}

\begin{theorem}\label{theorem_mfg}
Under Assumption \ref{mfgass}, the unique mean field equilibrium strategy is given by 
\begin{align}
    \phi^{*}_t &=   X_0 + (\widetilde I_0 - \widetilde S_0 \Delta_{0,T})  - (I_t-\Delta_{0,t} \widetilde S_0)
    \notag\\ &+ \int_0^t \Delta_{s,t}\frac{ \lambda  d\widetilde I_s+ \lambda d\overline X_s+d\widetilde S_t}{1+\lambda \Delta_{s,T}}
 + \int_0^t \widetilde \Delta_{s,t} \frac{\lambda d\check X_s}{1+
  \lambda\widetilde\Delta_{s,T}}.  \label{strathomo}
\end{align}
The equilibrium price has the following form:
\begin{equation}
\label{eq:price_incomplete_information}
P_t = S_t   - a(I_t-\Delta_{0,t} \widetilde S_0)
    + a\int_0^t \Delta_{s,t}\frac{ \lambda  d\widetilde I_s+ \lambda d\overline X_s+d\widetilde S_t}{1+\lambda \Delta_{s,T}}.
\end{equation}
\end{theorem}
The proof of Theorem \ref{theorem_mfg} follows the lines of that of Theorem \ref{theorem_complete} with some adjustments, and is thus omitted to save space.  

\subsection{Relationship between N-player setting and MFG setting}\label{CV}

{In this section, we study the relationship between the equilibrium strategies and prices in the $N$-agent market and those of the mean field game limit, and prove the following results.
\begin{itemize}
    \item The market price and the agent's strategy in the $N$-agent model converge to their respective mean field values as $N\to\infty$. This shows that to understand the behavior of agents and prices in the realistic $N$-agent market, one can use the mean-field game model, which does not require the knowledge of individual forecasts, but only that of the common information filtration.  
    \item An approximate equilibrium ($\varepsilon$-Nash equilibrium) in the $N$-player setting may be constructed from the MFG solution. In other words, an agent trading in the $N$-agent market may construct a strategy whose gain is sufficiently close to the optimal equilibrium gain using the mean-field game solution, which does not require the knowledge of the private forecasts of the other agents.  
\end{itemize}
}
To address these questions, we need to make more precise assumptions on the probabilistic setup of the problem. In particular, since we would like to study the convergence of the $N$-agent problem as $N\to \infty$, we consider an infinity of agents. In addition, all $N$-agent problems and the mean field problem must be defined on the same probability space. %In particular, $i$-th agent plays the role of the generic agent in the mean-field setting. 
\begin{assumption}${}$\label{samespace.ass}
\begin{itemize}
\item The process $S$  adapted to the filtration $\mathbb F^0$ and satisfies \eqref{intS}. 
\item The processes $(X^i)_{i=1}^\infty$ are square integrable $\mathbb F$-martingales.
\item There exists a square intergrable $\mathbb F$-martingale $\overline X$, such that for all $i\geq 1$, and all $t\in[0,T]$, almost surely, $\mathbb E[X^i_t|\mathcal F^0_t]=\overline X_t$.
\item The processes $(\check X^i)_{i=1}^\infty$ defined by $\check X^i_t = X^i_t - \overline X_t$ for $t\in[0,T]$, are orthogonal square integrable $\mathbb F$-martingales, such that the expectation $\mathbb E[(\check X^i_T)^2]$ does not depend on $i$. 
\end{itemize}
\end{assumption}
Let us fix $N<\infty$, and consider a market with $N$ agents. For a given $i\leq N$, we may define the "mean-field" strategy for the $i$-th agent as follows.
\begin{align}
    \phi^{MF,i*}_t&=   X^i_0 + (\widetilde I_0 - \widetilde S_0 \Delta_{0,T})  - (I_t-\Delta_{0,t} \widetilde S_0)
    \notag\\ &+ \int_0^t \Delta_{s,t}\frac{ \lambda  d\widetilde I_s+ \lambda d\overline X_s+d\widetilde S_t}{1+\lambda \Delta_{s,T}}
 + \int_0^t \widetilde \Delta_{s,t} \frac{\lambda d\check X^i_s}{1+
  \lambda\widetilde\Delta_{s,T}}  \label{mfstrat}
\end{align}
Unlike the true optimal strategy of the $i$-th agent, this strategy is computed using only the common information and the individual information of the $i$-th agent, it does not require the knowledge of the private forecasts of the other agents. Moreover, this strategy does not depend on $N$. The following two results show that, on the one hand, the true optimal strategy of the $i$-th agent in the $N$-player game converges to this mean-field strategy as $N\to \infty$, and on the other hand, that this mean-field strategy, if used by all agents in the $N$-player game, constitutes an $\varepsilon$-Nash equilibrium. Proofs of these results can be found in Appendix \ref{ProofCVNash}.

{
\begin{proposition}\label{CVNash}
Let Assumption \ref{samespace.ass} holds true, and let $\phi^{i*}$ denote the optimal position of the $i$-th agent in the $N$-player setting, given by \eqref{positionNash}, and by $\phi^{MF,i*}$ the optimal position in the mean field setting, given by \eqref{mfstrat}.
%begin{align*}
%\phi^{*}_t &= - I_t + \lambda\left[\Delta_{0,t}\frac{\widetilde I_0 +  \overline X_0}{1+\lambda \Delta_{0,T}}  + %\int_0^t \Delta_{s,t}\frac{d\widetilde I_s+ d\overline X_s}{1+\lambda \Delta_{s,T}}\right.\\&\left.+\widetilde %\Delta_{0,t}\frac{\check X^i_0}{1+\lambda\widetilde\Delta_{0,T}}+\int_0^t \widetilde \Delta_{s,t} \frac{ d\check X^i_s %}{1+
%  \lambda\widetilde\Delta_{s,T}} \right]. 
%\end{align*}
Then, for all $N\geq 1$, the differences between the strategy of a single agent, the aggregate strategy and the  equilibrium price in the $N$-agent model and the corresponding quantities in the mean-field model can be bounded as follows.
\begin{multline*}
\sup_{0\leq t\leq T}\mathbb E[(\phi^{i*}_t - \phi^{MF,i*}_t)^2] +\sup_{0\leq t\leq T}\mathbb E[(\overline\phi^{N*}_t - \overline \phi^*_t)^2]+\sup_{0\leq t\leq T}\mathbb E[(P^{N}_t - P_t)^2]\\ + \sup_{0\leq t\leq T}\mathbb E[(\dot{\overline\phi}^{N*}_t - \dot{\overline \phi}^*_t)^2]\leq \frac{C}{N^2} \mathbb E[\sup_{0\leq t\leq T} S_t^2] +  \frac{C}{N^2} \mathbb E[(\overline X_T)^2] + \frac{C}{N} \mathbb E[(\check X_T^i)^2],
\end{multline*}
where the constant $C$ depends only on the coefficients $\alpha$, {$b$}, $\lambda$ and $a$. 
\end{proposition}}

\begin{proposition} \label{propepsnashMF}
Under Assumption \ref{samespace.ass}, consider the vector of admissible strategies for the $N$-player game defined by equation \eqref{mfstrat} for $i=1,\dots,N$. Then, there is a constant $C<\infty$ which does not depend on $N$, such that for any other vector of admissible strategies $(\phi^i_t)^{i=1\dots,N}_{t\in [0,T]}$ for the $N$-player game,
$$J^{N,i}(\phi^i,\phi^{MF,-i*}) - \frac{C}{N^{\frac{1}{2}}} \leq J^{N,i}( \phi^{MF,i*},\phi^{MF,-i*}), \quad \forall i \in \{1,...,N\},\ \forall t \in [0,T].$$
In other words, the vector of strategies $(\phi^{MF,i*}_t)^{i=1\dots,N}_{t\in [0,T]}$ is an $\varepsilon$-Nash equilibrium for the $N$-player game with $\varepsilon = \frac{C}{N^{\frac{1}{2}}}$.  
\end{proposition}

\section{Intraday electricity prices: theoretical insights}\label{insights}
In this section we show through theoretical analysis how the main empirically documented features of electricity markets appear naturally as the result of our model. Empirical illustrations of these features are provided in the following section. We analyze the effect of market structure (number of participants, terminal penalty, trading costs and market impact parameters) on the overall costs/gains of participants as well as on the aggregate market parameters such as price volatility, the correlation between forecast and price and the impact of forecast adjustments on market prices. We consider the $N$-agent framework and make the following additional assumptions to simplify computations.
\begin{itemize}
    \item The fundamental price process $S$ is a martingale orthogonal to the forecast processes of the agents, with $\langle S\rangle_t = \sigma_S^2 t$. \item The trading cost parameter $\alpha$ is constant. 
    \item The forecast processes of agents satisfy
    $$
    \langle \overline X^N\rangle_t = \sigma_X^2 t,\quad  \langle \check X^i\rangle_t = \check\sigma_X^2 t,\quad \text{and}\quad \langle \overline X^N, \check X^i\rangle_t = 0\quad \forall\ i,
    $$ for some constants $\sigma_X$ and $\check \sigma_X$,
    where $\check X^i_t = X^i_t - \overline X^N_t$.
\end{itemize}
In addition, to make the notation more compact throughout this section we write $\tilde \alpha:= \alpha(1+b/2)$. 

Under these assumptions, the coefficients $\eta^{N}_{s,t}$ and $\Delta^{N}_{s,t}$ depend only on $t-s$ and not on $s$ and $t$ separately. We shall therefore write them as  $\eta^{N}_{t-s}$ and $\Delta^{N}_{t-s}$, and similarly for the other coefficients, from now and until the end of this section. The aggregate position of $N$ agents in equilibrium therefore writes:
\begin{align*}
  \bar \phi^{N}_t &=   \overline X^N_0  + \int_0^t \Delta^{N}_{t-s}\frac{ \lambda d\overline X^N_s}{1+\left(\frac{a}{N}+\lambda\right)  \Delta^{N}_{T-s}}.
\end{align*}

\subsection{Price impact of forecast adjustments}
Due to the non-storability of electricity, the prices of this commodity are strongly affected by demand and supply shocks. While for regular commodities these shocks may be compensated by changes in reserves, for electricity this is not possible. As a result, supply shocks caused, for instance, by the power plant breakdowns, and demand shocks often caused by weather forecast changes have a lasting impact on the price. In our model, the market impact of demand/supply shocks can be represented through a jump in the forecast process. An idiosyncratic supply shock may correspond to a jump in the individual forecast process $X^i$, while a generalized demand shock caused by weather forecast update may correspond to a jump in the aggregate forecast process $\overline X^N$. Consider for instance a jump $\Delta\overline X^N$ in the aggregate forecast occurring at time $t^*$.  The impact of this jump on the aggregate strategy $\bar \phi^N_t$, defined as the difference of the strategies with and without the forecast adjustment, is given by
$$
\delta \bar \phi^N_t = \mathbf 1_{t\geq t^*}\frac{\lambda\Delta^{N}_{t-t^*} \Delta \overline X^N}{1+\left(\frac{a}{N}+\lambda\right)\Delta^{N}_{T-t^*}}.
$$

\begin{figure}[h]
    \centering
    \includegraphics[width=0.5\textwidth]{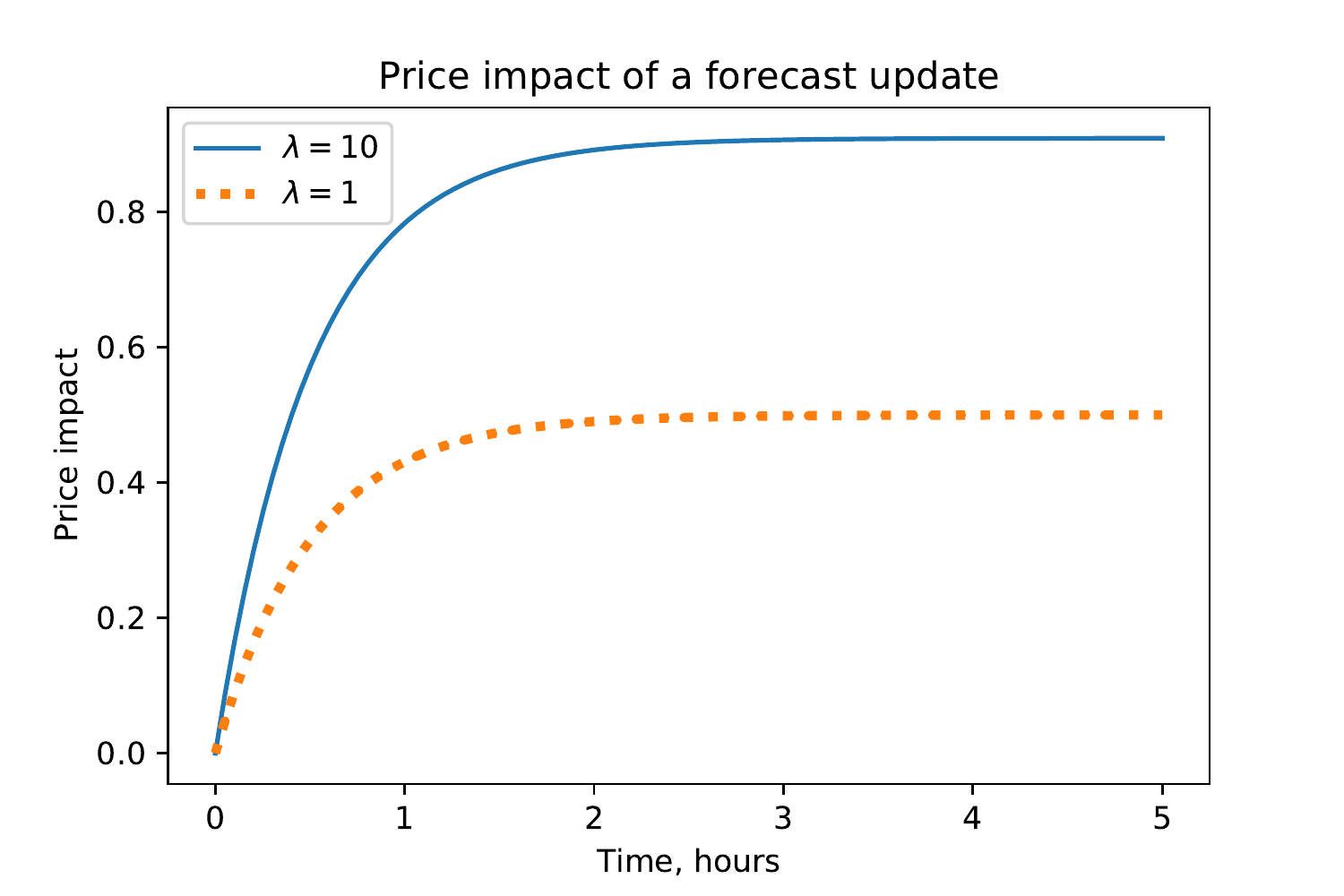}
    \caption{Price impact of an aggregate demand forecast adjustment  at $t=0$ (delivery time is $t=5$). Parameter values: $a=1$, $N=100$, $\alpha = 0.5$, $b=0$, forecast adjustment size: $1$ MW.}
    \label{impact.fig}
\end{figure}
For example, after a positive demand shock, the agents will need to purchase the missing electricity in the market, creating a permanent price impact given by a positive increasing continuous function of time (see Figure \ref{impact.fig}).

\subsection{Volatility and Samuelson's effect}\label{marketeffect}
We have seen that since the strategy $\bar \phi^N$ is differentiable, the quadratic variation of the equilibrium price $P^N_t$ coincides with the quadratic variation of the fundamental price. However, the actual observed volatility, which is estimated from discretely observed prices, may be different. The standard estimator of integrated variance by discrete quadratic variation over the interval $[t,t+h]$ with $M$ steps is given by
$$
Q^M(t,t+h) = \sum_{i=0}^{M-1} (P_{t+\frac{i+1}{M}h} - P_{t+\frac{i}{M}h})^2
$$
To focus on the average behavior of volatility rather than on individual random trajectories, we consider the expectation of this estimator. Finally, to estimate the expected instantaneous variance, it seems natural to consider this estimator over one time step and normalize it by the step size. Thus, the expression
$$
\frac{1}{h}\mathbb E[(P^N_{t+h} - P^N_t)^2]
$$
represents the average instantaneous price variance, estimated over time step $h$. 

The following lemma quantifies the behavior of this expression for small values of $h$. 
\begin{lemma} As $h\to 0$, the equilibrium price satisfies
\begin{align}
\frac{1}{h}\mathbb E[(P^N_{t+h} - P^N_t)^2] = \sigma_S^2  + h\mathcal V_t + O(h^2),\label{truevol}
\end{align}
where 
\begin{align}
\mathcal V_t  =  \frac{a^2 \lambda^2 \sigma_X^2}{\tilde\alpha^2} \int_0^t  \frac{(\eta^{N}_{t-s})^2}{(1+(\frac{a}{N}+\lambda)\Delta^{N}_{T-s})^2}ds.\label{expphi2}
\end{align}
\end{lemma}
\begin{proof}
The expected squared change in the price process satisfies:
\begin{align*}
    \mathbb E[(P^N_{t+h} - P^N_t)^2] &= \sigma_S^2 h + 2a\mathbb E[(S_{t+h}-S_t)(\bar \phi^N_{t+h}-\bar \phi^N_t)] + a^2\mathbb E[(\bar \phi^N_{t+h}-\bar \phi^N_t)^2]
\end{align*}
Since the fundamental process $S$ is orthogonal to the forecast process, the second term in the right-hand side above is zero. The third term satisfies:
\begin{align*}
\mathbb E[(\bar \phi^N_{t+h}-\bar \phi^N_t)^2] &= \sigma_X^2 \lambda^2 \int_0^t\frac{(\Delta^N_{t-s} - \Delta^N_{t+h-s})^2}{(1+(\frac{a}{N}+\lambda)\Delta^N_{T-s})^2}ds  \\ &+   \sigma_X^2 \lambda^2 \int_t^{t+h}\frac{(\Delta^N_{t+h-s})^2}{(1+(\frac{a}{N}+\lambda)\Delta^N_{T-s})^2}ds.
\end{align*}
From the explicit form of $\Delta^N_t$, it is clear that the second term above is of order of $O(h^3)$, and the first term equals 
$$
h^2\sigma_X^2 \lambda^2 \int_0^t\frac{((\Delta^N_{t-s})')^2}{(1+(\frac{a}{N}+\lambda)\Delta^N_{T-s})^2}ds,
$$
up to terms of order of $h^3$. 
\end{proof}
%\begin{align*}
%\mathbb E[(S_{t+h}-S_t)(\bar \phi^N_{t+h}-\bar \phi^N_t)] &= -\sigma_S^2 \int_{t}^{t+h} \frac{ \Delta^{N}_{t+h-s}ds}{1+\left(\frac{a}{N}+\lambda\right)  \Delta^{N}_{T-s}} \\ &= - \frac{\sigma_S^2}{\alpha(2+b)} \frac{ h^2}{1+\left(\frac{a}{N}+\lambda\right)  \Delta^{N}_{T-t}} + O(h^3)
%\end{align*}
As expected, as $h\to 0$, the expression \eqref{truevol} converges to the variance of the fundamental price $\sigma^2_S$. However, an agent using volatility estimator with time step $h$ on the fundamental price process, will find an extra variance of approximately $h\mathcal V_t$ (on average). For a given fixed time step, the function $\mathcal V$ can thus be used as a proxy of the additional volatility of the equilibrium prices. 

%Since the market impact component of the price is given by $a\bar \phi^N_t$, we may use the expectation of the squared derivative of the aggregate strategy, $\mathbb E[(\dot{\bar \phi}^N_t)^2]$, as a measure of variability of the aggregate strategy, and thus as a proxy for the additional variance of the equilibrium price. 

\begin{figure}[h]
    \centering
    \includegraphics[width=0.5\textwidth]{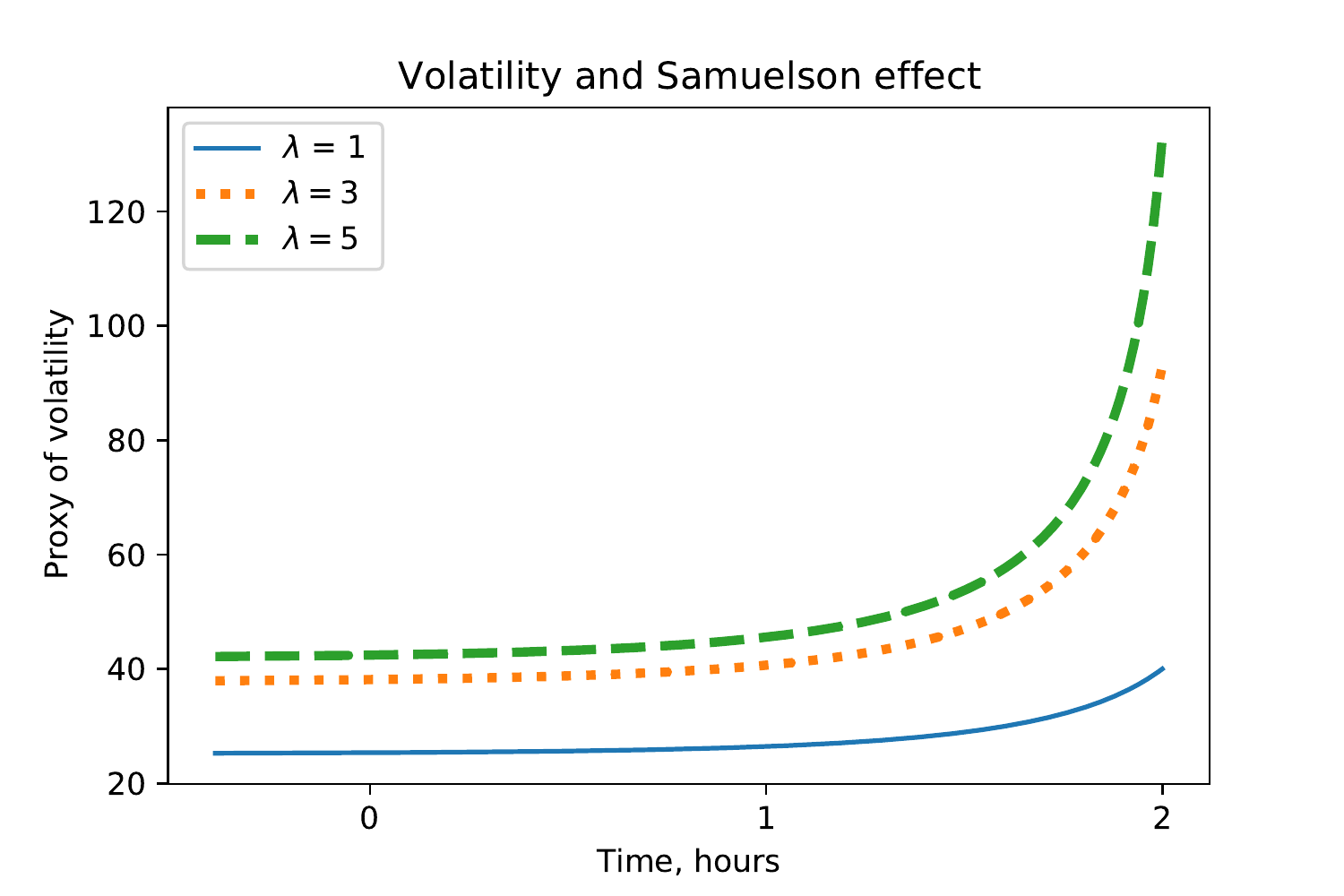}
    \caption{Proxy of expected observed volatility as function of time (delivery time is $t=2$). Parameter values: $a=1$, $N=100$, $\alpha = 0.5$, $b=0$, $\sigma_X = 50$.}
    \label{vol.fig}
\end{figure}
In this section we draw conclusions about the behavior of price volatility by analyzing this proxy, and in Section \ref{empirical} we will show in numerical examples that the actual volatility, estimated from discrete observations of simulated market price exhibits similar behavior. 
%Under the assumption of this section, this gives:
%\begin{align}
%\dot{\bar\phi}^{N}_{t} &=  \frac{\eta^{N}_{t}}{\alpha{\left(1+\frac{b}{2}\right)}}\frac{- S_{0}  + \lambda \overline X^N_{0}}{1+\left(\frac{a}{N}+\lambda\right) \Delta^{N}_{T}}  + \int_{0}^{t} \frac{\eta^{N}_{t-s}}{\alpha{\left(1+\frac{b}{2}\right)}}\frac{ -dS_s+ \lambda d\overline X^N_s}{1+\left(\frac{a}{N}+\lambda\right)  \Delta^{N}_{T-s}}.\notag\\
%\mathbb E[(\dot{\bar \phi}^N_t)
%^2]& = \frac{(\eta^{N}_{t})^2}{\alpha^2{\left(1+\frac{b}{2}\right)^2}}\frac{( S_{0}  - \lambda \overline X^N_{0})^2}{(1+\left(\frac{a}{N}+\lambda\right) \Delta^{N}_{T})^2}\notag\\
%&+ \frac{\sigma_S^2 + \lambda^2 \sigma_X^2}{\alpha^2{\left(1+\frac{b}{2}\right)^2}} \int_0^t  \frac{(\eta^{N}_{t-s})^2}{(1+(\frac{a}{N}+\lambda)\Delta^{N}_{T-s})^2}ds. \label{expphi2}
%\end{align}

\begin{itemize}
\item First of all, since $\Delta^N_t$ is increasing in $t$, the observed price volatility increases at the approach of the delivery date in our model (see Figure \ref{vol.fig}). This phenomenon, well documented in electricity futures and other futures markets \cite{jaeck2016volatility} is known as the Samuelson effect and we also illustrate it empirically in Section \ref{empirical}. 
\item The observed price volatility is increasing in $\lambda$: stronger imbalance penalties lead to higher volatility in the intraday market (see Figure \ref{vol.fig}). Moreover, 
$$
\lim_{\lambda\to \infty} \mathcal V_t = \frac{a^2\sigma_X^2}{\tilde\alpha^2}\int_0^t \frac{(\eta^N_{t-s})^2}{(\Delta^N_{T-s})^2}ds,
$$
and the latter expression explodes for $t\to T$. Thus, we conclude that the Samuelson effect is also stronger for higher imbalance penalties. 
\item In the small liquidity cost regime ($\tilde\alpha\to 0$), for $0<t<T$, $\Delta^N_{t-s}\to \frac{N}{(N-1)a}$ uniformly on $s\in[0,t]$. Therefore, for $N>1$, 
$$
\mathcal V_t \sim \frac{a^2\lambda^2 \sigma_X^2}{\tilde\alpha^2} \frac{(a+\lambda)^2 N^2}{a^2(N-1)^2} \int_0^t (\eta^N_{t-s})^2 ds \sim \frac{\lambda^2 \sigma_X^2}{\tilde\alpha} \frac{(a+\lambda)^2}{2a} \frac{N}{N-1}
$$
 This shows that with decreasing trading costs extra variance of the equilibrium price grows like $\frac{1}{\tilde\alpha}$.  Lower transaction costs allow the agents to follow the forecasts more closely, leading to a higher volatility of the aggregate position and of the market price.  On the other hand, since the function $N\mapsto \frac{N}{N-1}$ is decreasing in $N$, we conclude that price volatility in the small liquidity cost regime is decreasing with the number of agents: in our model, competition between agents increases market frictions and leads to reduced volatility.

    \item In the large liquidity cost regime ($\tilde\alpha\to \infty$), $\eta^N_t\to 1$ and $\Delta^{N}_t \sim \frac{t}{\tilde\alpha }$, so that 
    $$
    \mathcal V_t \sim \frac{a^2\lambda^2 \sigma_X^2 t}{\tilde \alpha^2}. 
    $$
    Higher liquidity costs decrease the trading rate of agents and lead to a lower overall market volatility.
%    \item  {\color{blue}Large penalty limit: $\lambda\to \infty$. In this case, the extra variance converges to a limit: 
%\begin{align}
%\mathcal V_t  \to \frac{a^2(\eta^{N}_{t})^2}{\tilde\alpha^2}\frac{ (\overline X^N_{0})^2}{(\Delta^{N}_{T})^2}+ \frac{a^2\sigma_X^2}{\tilde\alpha^2} \int_0^t  \frac{(\eta^{N}_{t-s})^2}{(\Delta^{N}_{T-s})^2}ds
%\end{align}

%    Higher imbalance penalties force the agents to follow the forecasts more closely and lead to an overall higher price volatility. 
    
%    \todo[inline]{ jai mis le resultat 'logique' mais je suis pas arrivée a remontrer la limite}}
\end{itemize}

%We shall now analyze the variability of the aggregate strategy just before delivery ($t\to T$). In this case, the explicit computation of the integral gives,
%\begin{align*}
%\mathbb E[(\dot{\bar \phi}^N_T)
%^2]& = \frac{(\eta^N_{T})^2}{\alpha^2}\frac{( S_{0}  - \lambda \overline X^N_{0})^2}{(1+\left(\frac{a}{N}+\lambda\right) \Delta^N_{T})^2} + \frac{\sigma_S^2 + \lambda^2 \sigma_X^2}{\alpha} \int_0^{\Delta^N_T}  \frac{1-a(N-1)u/N}{(1+(\frac{a}{N}+\lambda)u)^2}du\\
%& = \frac{(\eta^N_{T})^2}{\alpha^2}\frac{( S_{0}  - \lambda \overline X^N_{0})^2}{(1+\left(\frac{a}{N}+\lambda\right) \Delta^N_{T})^2} + \frac{\sigma_S^2 + \lambda^2 \sigma_X^2}{\alpha} \frac{(a+\lambda) N^2}{(a+\lambda N)^2}(1-\frac{1}{1+\Delta^N_T(a/N+\lambda)} ) \\
%&- \frac{\sigma_S^2 + \lambda^2 \sigma_X^2}{\alpha} %\frac{a(N-1)N}{(a+\lambda N)^2}  %\log(\Delta^N_T(a/N+\lambda)+1)
%\end{align*}

%The derivative of the aggregate position is given by
%$$
%\dot{\Bar{\phi}}^{N}_t =   \int_{-\infty}^t \frac{  \eta^N_{t-s}}{\alpha  (1+(a/N+\lambda)\Delta^N_{-s})}(\lambda d\overline X^N_s - dS_s),\quad t<0. 
%$$

\subsection{Price-forecast covariance} To understand how the forecast updates influence prices, we compute the covariance of the increment of the aggregate strategy over an interval of length $h$ with the increment of the aggregate forecast over the same interval.  
Using the explicit form of the strategy, we easily obtain,
\begin{align*}
\text{Cov}[\bar \phi^N_{t+h} - \bar \phi^N_{t},\overline X^N_{t+h}  - \overline X^N_{t}] = \lambda (\sigma_X)^2 \int_{t}^{t+h} \frac{\Delta^{N}_{t+h-s}}{1+(\frac{a}{N}+\lambda) \Delta^{N}_{T-s}}ds. 
\end{align*}
From this expression, we conclude that the covariance of equilibrium price with forecast updates increases when the terminal penalty $\lambda$ increases, and when the time $t$ approaches the delivery date . 
\subsection{Trading costs}
In our model, the agents face three types of costs: the trading costs, the market impact costs, and the balancing costs. Using the martingale property of $S$ and other assumptions of this section, these costs are evaluated as follows: 
\begin{align*}
C^{N,i}_{tra} &=    \mathbb{E}\left[\int_{0}^{T}\frac{\alpha}{2}\Dot{\phi^i_t}{\left(\dot \phi^i_t + b \dot{\bar\phi}^{N,-i}_t\right)}dt\right]\\ &= \frac{\alpha(1+b))}{2}\mathbb{E}\left[\int_{0}^{T}({\Dot{\bar\phi}^N_t})^{2}dt\right] + \frac{\alpha}{2}\frac{N-1-b}{N-1}\mathbb{E}\left[\int_{0}^{T}({\Dot{\check\phi}^N_t})^{2}dt\right]\\
C^{N,i}_{imp} &= \mathbb{E}\left[\int_{0}^{T} a\Dot{\phi^i_{t}}(\bar \phi^N_t - \bar \phi^N_0)dt\right] = \frac{a}{2}\mathbb E\left[(\bar \phi^N_T-\bar \phi^N_0)^2\right]\\
C^{N,i}_{bal} &= \frac{\lambda}{2}\mathbb{E}\left[(\phi^i_{T}- X^i_{T})^2\right] = \frac{\lambda}{2}\mathbb{E}\left[(\bar\phi^N_{T}- \overline X^N_{T})^2\right]+\frac{\lambda}{2}\mathbb{E}\left[(\check\phi^i_{T}- \check X^i_{T})^2\right]
\end{align*}
After some tedious but straightforward computations, these costs are found to have the following integral form:
\begin{align*}
C^{N,i}_{tra} &= \frac{1+b}{1+\frac{b}{2}}\frac{ \lambda^2 \sigma_X^2}{4}\int_0^T dt \frac{(1+\eta^{N}_{T-t})\Delta^{N}_{T-t}}{ (1+(\frac{a}{N}+\lambda)\Delta^{N}_{T-t})^2}\\&+\frac{N-1-b}{N-1} 
\frac{\lambda^2 \check\sigma_X^2}{4}\int_0^T dt  \frac{ (\tilde\eta^{N}_{T-t}+1)\widetilde \Delta^{N}_{T-t} }{ (1+(\frac{a}{N}+\lambda)\widetilde\Delta^{N}_{T-t})^2}
\\C^{N,i}_{imp} &=   \frac{a \sigma_X^2 \lambda^2}{2}\int_0^T \frac{(\Delta^{N}_{T-t})^2}{ (1+(\frac{a}{N}+\lambda)\Delta^{N}_{T-t})^2}dt
\\C^{N,i}_{bal} &= \frac{\lambda \sigma_X^2}{2} \int_0^T \frac{(1+\frac{a}{N}\Delta^{N}_{T-t})^2}{ (1+(\frac{a}{N}+\lambda)\Delta^{N}_{T-t})^2}dt+\frac{\lambda\check\sigma_X^2}{2} \int_0^T \frac{(1+\frac{a}{N}\widetilde\Delta^{N}_{T-t})^2}{ (1+(\frac{a}{N}+\lambda)\widetilde\Delta^{N}_{T-t})^2}dt,
\end{align*}
which leads to the following conclusions:
\begin{itemize}
    \item Trading costs are proportional to forecast variances: more precise forecasts lead to lower trading costs. However, while the trading costs and the balancing costs depend both on the volatility of aggregate forecast and that of the individual forecast, the market impact costs only depend on the volatility of the aggregate forecast. Thus, an agent who has a better individual forecast will pay lower trading and balancing costs but the same market impact costs. 
    \item Since both $C^{N,i}_{tra}$ and $C^{N,i}_{imp}$ are increasing in $\lambda$, stronger imbalance penalties lead to higher trading and market impact costs. The balancing cost $C^{N,i}_{bal}$ is increasing in $\lambda$ for small values of $\lambda$, but may become decreasing for large $\lambda$. When $\lambda\to \infty$, the market impact costs and the balancing costs remain bounded, however it can be shown that the single agent trading cost tends to $+\infty$ at the rate of $\log \lambda$, thus very high imbalance penalties lead to prohibitive trading costs and are therefore detrimental for market liquidity. 
    \item In the case of small liquidity costs ($\alpha\to 0$), each component of the cost converges to a nonzero limit:
    \begin{align*}
    C^{N,i}_{tra} &\to \frac{1+b}{1+\frac{b}{2}}\frac{ \lambda^2 \sigma_X^2T}{4}\frac{(N-1)a}{N(a+\lambda)^2} + \frac{N-1-b}{N-1} \frac{\lambda^2\check \sigma_X^2T} {4} \frac{aN T}{(a+\lambda N)^2}\\
    C^{N,i}_{imp} & \to \frac{a\sigma^2_X  \lambda^2 T}{2(a+\lambda^2)} \\
    C^{N,i}_{bal} & \to \frac{\lambda \sigma_X^2 T}{2} \frac{a^2}{(a+\lambda)^2} + \frac{\lambda \check\sigma_X^2 T}{2} \frac{a^2}{(a+\lambda N)^2}.
    \end{align*}
    As the cost per trade decreases, the agents trade more actively so that the overall trading cost does not tend to zero. 
\end{itemize}

\section{Empirical results and numerical illustrations}\label{empirical}

In this section our objective is to  analyze the empirically observed features of intraday market prices, demonstrate that these features are reproduced by our model, and illustrate other properties of our model, such as the convergence of the $N$-agent model to the mean-field limit, with numerical examples.

\subsection{Stylized features of intraday electricity market prices} \label{stylized}
\paragraph{A brief description of our dataset}
% The EPEX intraday electricity market (\texttt{www.epex.com}) opens every day at 3 p.m and allows to trade all delivery hours of the following day up to 5 minutes before delivery. Each delivery hour is a distinct product in the market. Different European geographic zones are available and in this study we focus on the German area. It is also possible to trade in quarter-hours, but in the empirical study we focus on the full hours only. 
To compute the empirical price analyzed in the following sections, we used the limit order book data provided by EPEX electricity market for the Germany delivery zone for the $1^{st}$ quarter of 2015 and January 2017. Although in the market it is possible to trade in quarter-hours, in this study we focus on the full hours only. The dataset contains full information about sell and buy orders recorded on any given day, whether they result in a transaction or not. From this data we reconstruct the state of the order book, which allows us in turn to derive the mid-quote price and the bid-ask spread.
%\todo{Add a brief description of market operation, of LOB data and explain how we reconstruct mid-quote prices (quelques lignes)}
\paragraph{Market liquidity}
In Figure \ref{Fig0}, we plot the distribution of the times of orders and transactions as function of time to delivery computed over all orders and transactions in February 2015. We observe that the liquidity starts to appear only 5-6 hours before delivery, and grows very quickly at the approach of the delivery date. 
\begin{figure}[h]
      \centering
      \caption{Distribution of times orders and transactions in February 2015}
      \label{Fig0}
      \includegraphics[width =0.80\textwidth]{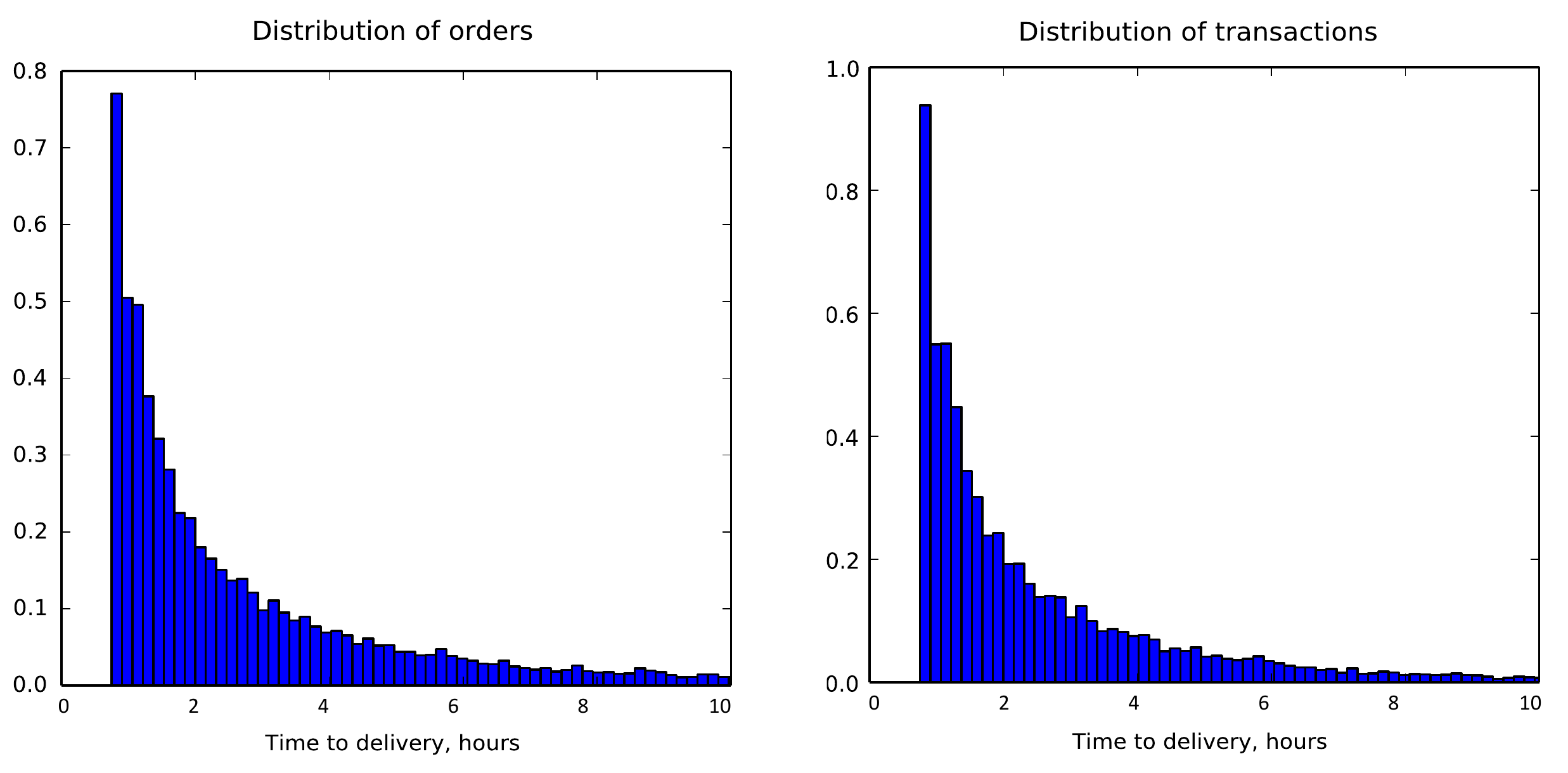}
\end{figure}
{ We also performed an estimation of the  bid-ask spread in the German intraday electricity market for January 2017 for different delivery times. For each  delivery time, Figure \ref{FigSpread}, shows that the spread averaged over each hour and over all days in January 2017 decreases as we approach the end of the trading period.} This is consistent with the assumption that the market is used by the renewable energy producers to adjust their positions when precise forecasts become available. 
\begin{figure}[h]
      \centering
      \caption{Average spread per hour as a function of time to delivery over January 2017}
      \label{FigSpread}
      \includegraphics[width =0.5\textwidth]{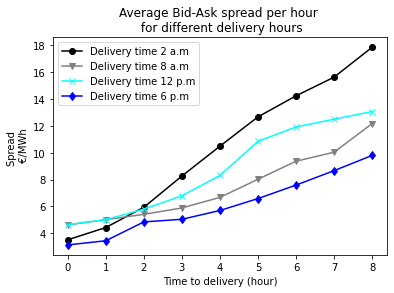}
\end{figure}

\paragraph{Price volatility}
To estimate the empirical volatility, we consider mid-quote prices reconstructed from the limit order book data of the Germany delivery zone for January 2017, as explained above. The mid-quote price was computed on a uniform grid with a time step of 1 minute. In January 2017 the market was already relatively liquid: the average number of daily price changes for a given delivery hour varied between approximately 3400 for the least liquid delivery hour (2 AM to approximately 5800 for the most liquid delivery hour (6 PM). Given that, as we observed above, liquidity is concentrated in the last 5-6 hours, a one-minute interval during this time contains many price changes and the market microstructure effects are limited.

The observed midquote price is denoted $(\tilde P_t)_{t \in [0,T]}$. %For the purpose of estimating volatility, we assume it has the following dynamics:
%\begin{equation}\label{mq}
%    d\tilde P_t = \mu_t dt + \sigma_t d W^{\tilde{P}}_t, \quad 0 \leq t\leq T,
%\end{equation}
%where $(W^{\tilde{P}}_t)_{t \in [0,T]}$ is a Brownian motion, and $(\mu_t)_{t\in [0,T]}$ and $(\sigma_t)_{t \in [0,T]}$ are adapted processes. 
%The volatility is an important feature of the market. We show that the volatility observed empirically in the intraday market can be explained by the impact of the uncertain production of intermittent energies. We used as previously the observed intraday midquote price $(\tilde P_t)_{t\in [0,T]}$ derived from the limit order book data of January 2017 to estimate the empirical volatility of the market price for each delivery time. \\
%We are especially interested in estimating the instantaneous volatility $(\sigma_t)_{t \in [0,T]}$ of $(\tilde P_t)_{t\in [0,T]}$, introduced in the dynamics \eqref{mq}. 
We denote by $n$ the number of observations in the data of January 2017 and by $\{t_0,\dots,t_i, \dots, t_n\}$ the (uniform) time grid over which the observations are available. In contrast with the integrated volatility whose estimator is generally given by $\widehat{\int_0^T \sigma^2_sds}= \sum_{i=1}^n \Delta \tilde P^2_{t_{i-1}}$, estimating the instantaneous volatility is less straightforward. 
Following \cite{kristensen2010nonparametric}, we use a kernel-based non parametric estimator of the instantaneous volatility:
\begin{equation}\label{volestimation}
    \hat \sigma_t^2 = \frac{\sum_{i = 1}^n K_h(t_{i-1}-t)\Delta \tilde P^2_{t_{i-1}}}{\sum_{i = 1}^n K_h(t_{i-1}-t)(t_{i}-t_{i-1})},
\end{equation}
where $K(.)$ is the Epanechnikov kernel: $K(x) = \frac{3}{4} (1-x^2) \mathbbm{1}_{[-1,1]}(x)$
and $K_h(x)= \frac{1}{h}K(\frac{x}{h})$. The parameter $h$ was taken equal to $0.08$ hour ($\approx 5$ minutes) after performing some cross-validation analyses and sensitivity tests. The paths of the estimated volatility as function of time to delivery for different delivery hours are given in Figure \ref{FigVol}. We observe that the volatility increases as delivery time draws near and market participants trade more actively, giving an empirical evidence of the presence of the Samuelson effect in electricity market. 

\begin{figure}[h]
      \centering
      \caption{Instantaneous market volatility for different delivery hours}
      \label{FigVol}
      \includegraphics[width =0.5\textwidth]{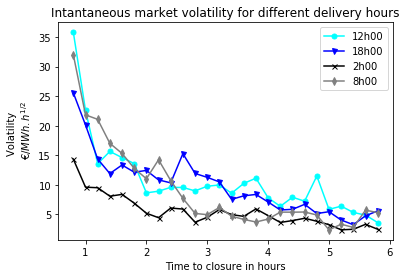}
\end{figure}
%\todo{Change the legend: instantaneous volatility for different delivery hours}

%\todo{Mettre les unités pour la vol sur ce graphique et sur le graphique de la vol simulé}
\paragraph{Correlation between price and renewable indeed forecasts}
We finally study the empirical correlation between the intraday market prices and the renewable wind production forecasts. Unlike the rest of the paper, here we use actual wind infeed forecasts, not the demand forecasts. To compute empirical correlation estimates, we use the limit order book data from the intraday EPEX market of the  first three months of 2015 for the Germany delivery zone, from which, as before, we compute the mid-quote prices. The production forecasts correspond to the same period and are updated every 15 minutes for each delivery hour. In Figure \ref{FigPearsonempirique}, we plot the correlation between the increments of the market price and the increments of the production forecasts for the delivery time 12h (averaged over 90 days in the dataset), together with the 2-standard deviation bounds. To match the forecast update frequency, the mid-quote price is also sampled at 15-minute intervals here.

\begin{figure}[h]
      \centering
      \caption{Correlation between the market price increments and the renewable production forecast increments for the German delivery zone in winter 2015}
      \label{FigPearsonempirique}
      \includegraphics[width =0.6\textwidth]{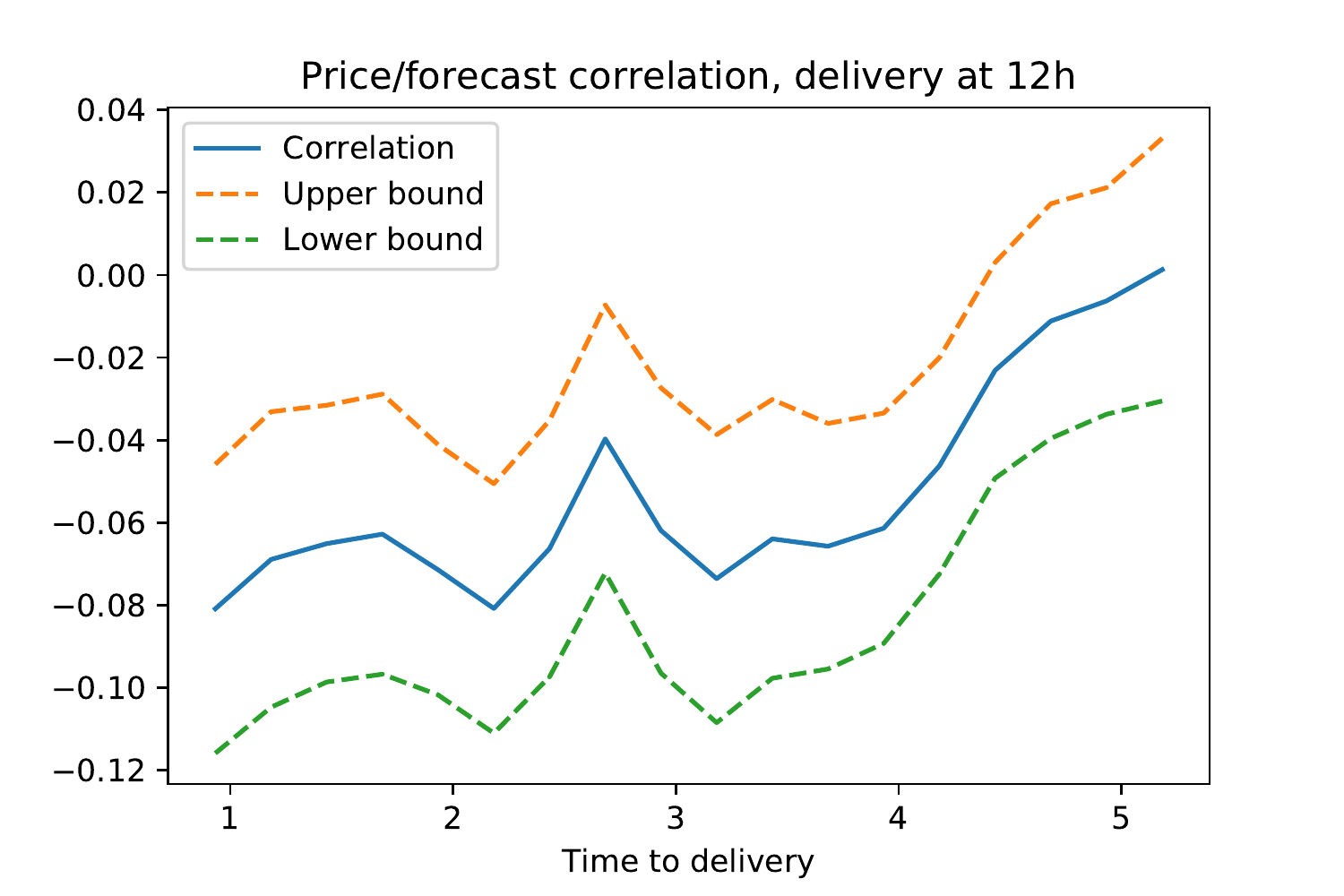}
\end{figure}

We find that the correlation between the price increments and those of the production forecast is negative and increases in absolute value as we approach the delivery time.  
\subsection{Numerical illustration of our model}
\paragraph{Model specification}
We now define the dynamics for the fundamental price and for the demand forecasts used in the simulations. We also give the chosen values of the different parameters. Our objective here is to illustrate the features of the model and show that it reproduces the stylized facts of the market prices. Therefore, the majority of the parameters are not precisely estimated, but are given plausible values. 

The evolution of the fundamental price is described as follows: 
\begin{equation}
    dS_t = \sigma_S dW_t
\end{equation}
where $\sigma_S$ is a constant and $(W_t)_{t \in [0,T]}$ is Brownian motion. % is a $\mathbb{F}$ (resp.$\mathbb{F}^0$)-Brownian motion in the complete (resp. incomplete) information setting. 
We also assume that the liquidity function $\alpha(.)$ is given by
\begin{eqnarray}
\label{alpha}
    \alpha(t) & = & \alpha \times (T-t)+ \beta, \quad \forall t \in [0, T] 
\end{eqnarray}
where $\alpha$ and $\beta$ are strictly positive constants.  The liquidity function is decreasing with time. This assumption relies on the fact that, as we observed in Section \ref{stylized}, the market becomes more liquid as we get closer to the delivery time and it is less costly to trade when the market is liquid. 
%The liquidity function for the major agent $\alpha_0$ has the same form.

To simulate demand forecasts, we assume the following dynamics: 
\begin{eqnarray}
 d \bar X^N_{t} & = & \sigma_X d\bar B_t \\
 d \check X^i_{t} & = & \check\sigma_X dB^i_{t}, \qquad i \in \{1, \dots, N\}
\end{eqnarray} 
where $\sigma_X$ and $\check\sigma_X$ are  constants and $(\bar B_t)_{t \in [0,T]}$, $(B^i_{t})_{t \in [0,T]}$ are independent Brownian motions, also independent from $(W_t)_{t \in [0,T]}$.
%(resp.$\mathbb{F}^0$)-Brownian motions in the complete (resp. incomplete) information setting.$\mathbb{F}$-Brownian motions in the complete information setting (respectively $\mathbb{F}^0$ and $\mathbb{F}$ or $\mathbb{F}^i$ adapted processes and Brownian motions in the incomplete information setting). The different Brownian motions entering the problem are assumed to be independent. 

In this illustration, we choose the same parameters for the dynamics of the common and the individual demand forecasts (that is, $\sigma_X = \check \sigma_X$). The common volatility is calibrated to wind energy forecasts in Germany %for Germany %supplied by EWE Trading GmbH 
over January 2015 during the last quotation hour, by using the classical volatility estimator %{\color{magenta} What about the other forecasts? The price may depend on them in the N-agent setting.}
\begin{equation}\label{volforecast}
   \sigma_X = \check{\sigma}_X = \frac{\sqrt{\Delta t}}{n' - 1}\sum_{i=1}^{n'}Y_i^2
\end{equation}
with $\Delta t$ the time step between two observations, $Y_i = X_{t_i} - X_{t_{i-1}}$ the increment between two successive observations and $n'$ the total number of observed increments. As the forecasts are updated every 15 minutes, there are three increments during the last trading hour, available on each day from the 3$^{rd}$ of January to the 31$^{th}$ of January. Thus, for each delivery hour we dispose of $n'= 87$ increments points to estimate the volatility.}  
The volatility, as well as the other model parameters are specified in Table \ref{parameters}.

\begin{table}[h]
\centering
   \begin{tabular}{ |c|c||c|c|}
     \hline
     Parameter & Value & Parameter & Value \\ \hline
     $S_0$ & $40$ \euro{}/MWh &  $a$ & $1$  \euro{}/MWh$^{2}$  \\ \hline
     $\sigma_{S}$ & $10$ \euro{}/MWh$\cdot$h$^{1/2}$ & $\lambda$ & $100$ \euro{}/MWh$^{2}$\\ \hline
     $\overline{X}_0, \, \check X^{i}_0$ & 0 MWh & $N$ & $100$ \\ \hline
      $\sigma_X$,$\check\sigma_{X}$ &  $73$ MWh/h$^{1/2}$ &  $\alpha$ & 0.14 \euro{}/MW$^{2}$$\cdot$h \\ \hline
     ${b}$ & {0}  & $\beta$ & 0.06 \euro{}/MW$^{2}$ \\ \hline
   \end{tabular}
   \caption{Parameters of the model}
   \label{parameters}
 \end{table}
\paragraph{Price trajectories.}
In Figure \ref{Fig2}, we plot a simulated trajectory of the fundamental price $S$ %(S_t)_{t\in [0,T]}$
starting six hours before the delivery time (corresponding to $t=0$), up to the time $T$ of delivery, together with the market price $P$  associated with the different  settings studied in this paper: the $N$-player Nash equilibrium with $N=100$ players, the mean field and the $\epsilon$-Nash equilibrium. Graphs were all simulated with the same demand forecasts, initial values, volatilities and parameters as specified in Table \ref{parameters}.

\begin{figure}[h]
      \centering
      \caption{Model price trajectories (left) associated to a given common demand forecast trajectory (right) in different settings}
      \label{Fig2}
      \includegraphics[width =0.48\textwidth]{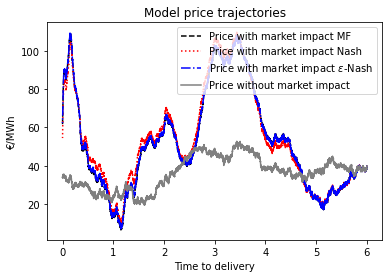}
    \includegraphics[width =0.48\textwidth]{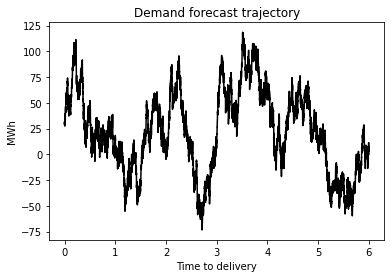}
\end{figure}
%\todo{Change the legends: left: Model price trajectories, Right: Demand forecast trajectory}
In all settings, the model reflects the price impact of the positions taken by the agents.  This price impact is influenced by the market price and the demand forecasts.
If agents anticipate to have overestimated the demand (negative values of the demand  process), there is an excess of supply in the market, thus the price impact is negative and the market price decreases. On the contrary, if they anticipate to have underestimated the demand  (positive values of the demand forecast process), there is a lack of supply and the market price increases. %\todo{Can we explain the difference between the $\varepsilon$-Nash and the other two trajectories?}
 
{In Figure \ref{Fig2}, the parameters are as specified in Table \ref{parameters}, with $b=0$. To emphasize the impact of crowd behavior on the cost of trading and the strategy, we display in Figure \ref{Figb} the  trajectory of the mean field at the equilibrium for different values of the parameter $b$. This is sufficient to capture the effect of a possible synchronization between agents since, from the discussion in Section \ref{Completeinfo}, only the common part of the strategy is impacted by the effective cost $\alpha(1+b/2)$.}

\begin{figure}[h]
      \centering
      \caption{Mean-field trajectories and demand forecasts as function of the impact of the crowd on the cost of trading }
      \label{Figb}
      \includegraphics[width =0.48\textwidth]{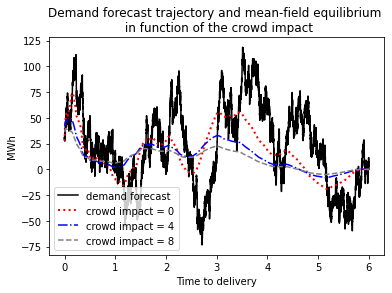}
\end{figure}
{As we saw in Section \ref{Completeinfo}, higher crowd trading parameter $b$ leads to increased trading costs and therefore reduced volatility of the aggregate strategy, which therefore follows the forecast updates less closely.  }

\paragraph{Volatility and correlation}

In this paragraph we compare the price volatility and the correlation between price and renewable infeed forecasts in our model with the empirical ones. We have already seen through theoretical analysis in Section \ref{marketeffect} that our model reproduces the observed features of the volatility; the goal of this paragraph is to confirm this using simulated prices. We once again highlight the fact that the market impact can induce an increase in the price variations but no changes in the quadratic variation since the price impact, though it is stochastic, has a finite quadratic variation.
However, the volatility estimated from discrete price observations, which is the only quantity relevant in practice, does increase in our model, as we shall see below. 

We focus on hourly products and on several different delivery hours: 2 AM, 8 AM, 12 PM and 6 PM to include both peak (high electricity demand) and off-peak (low electricity demand) times.
The volatility of the fundamental price $S$ is assumed to be constant, ($\sigma_S = 10$ \euro{}/MWh$\cdot$h$^{1/2}$) to ensure that the observed volatility changes are only due to the stochastic drift of the market price, i.e., the aggregate trading rates of the agents. 
The volatility of the production forecasts for the different delivery hours has been calibrated using the estimator defined in  \eqref{volforecast}  and is shown in Table \ref{voltab}.

\begin{table}[h]
\centering
   \begin{tabular}{ |c|c|}
     \hline
     Hour & Volatility (MWh/h$^{1/2})$ \\ \hline
     2h00 & $67$  \\ \hline
     8h00 & $81$\\ \hline
     12h00 &$73$   \\ \hline
     18h00 &  $73$  \\ \hline
   \end{tabular}
   \caption{Calibrated volatility of the production forecast for different delivery hours}
   \label{voltab}
 \end{table}

During peak hours, both market activity and liquidity are higher. To account for this phenomenon in our model, we chose different levels of the liquidity coefficients $\alpha$ and $\beta$  defined in \eqref{alpha} and presented in Table \ref{tabliquidity}. 
\begin{table}[h]
\centering
   \begin{tabular}{ |c || c | c | }
     \hline
     \backslashbox{Hours}{Coefficients} &$\alpha$ (\euro{}/h.MW$^{2}$) & $\beta$ (\euro{}/MW$^2$) \\ \hline
     2h00 & 0.24 & 0.10 \\ \hline
     8h00 & 0.10 & 0.04 \\ \hline
     12h00 & 0.06 & 0.02 \\ \hline
     18h00 & 0.14 & 0.06 \\
     \hline
   \end{tabular}
   \caption{Liquidity coefficients used for different delivery hours}
   \label{tabliquidity}
 \end{table}
 
Since calibrating the model to market data is not the purpose of this study, we chose plausible values for these coefficients in an ad hoc manner with lower trading costs corresponding to delivery hours for which the market is more liquid. All other model parameters are specified in Table \ref{parameters}. 

\begin{figure}[h]
      \centering
      \caption{Simulated model volatility for different delivery hours}
      \label{Fig3}

    \includegraphics[width =0.50\textwidth]{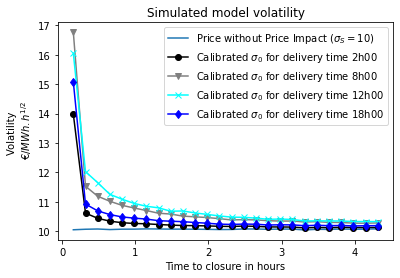}

\end{figure}

Figure \ref{Fig3} shows the estimated volatility of the simulated  model price $P$ in the Nash $N$-player game setting with $N=100$, averaged over 1000 simulations. The volatility was computed using the estimator \eqref{volestimation}, with the same window width and time step as in the empirical analysis.  From this graph we can see that the model is able to reproduce the increasing shape of the empirical market price volatility at the approach of the delivery time, and that it captures the different levels of volatility corresponding to the different delivery hours. 

%From the point of view of the intermittent energy producers, this can be explained by the fact that they aim to be as close as possible to the realized production at the delivery time and thus readjust  their positions more frequently as the delivery time approaches. This provides a theoretical explanation of the impact of renewable energies on the electricity intraday market features.

%\todo{Change the legend: simulated model volatility}
\paragraph{Correlation between price and renewable infeed}

An important stylized feature of intraday market prices, observed empirically in \cite{kiesel2017econometric} is the correlation between the price and the renewable production forecasts. Figure \ref{FigPearson} plots the correlation between 15-minute increments of the simulated market price and the 15-minute increments of the simulated renewable production forecasts as function of time. For each time step, the correlation $\rho_t = \text{corr}(\Delta Y_t, \Delta P_t)$ is computed by Monte Carlo using the following estimator:
{
$$\hat \rho_{t} = \frac{\sum_{k =1}^{N_{sim}} (\Delta Y^k_t - \overline{\Delta Y}_t)(\Delta P^k_t - \overline{\Delta P}_t)}{\sqrt{\sum_{k = 1}^{N_{sim}} (\Delta Y^k_t - \overline{\Delta Y}_t)^2\sum_{k = 1}^{N_{sim}}(\Delta P^k_t - \overline{\Delta P}_t)^2}},$$
with $N_{sim}$ stands for number of simulations (we considered $N_{sim} = 50000$), {$\Delta Y^k_t =-(\overline{X}^{N,k}_{t+dt}-\overline{X}^{N,k}_{t})$}, $\Delta P^k_t = P^{N,k}_{t+dt}-P^{N,k}_{t}$} and $N = 100$. Notice that we use the minus sign in front of the forecast increment to plot the correlation of production forecasts, whereas $\overline X$ stands for the demand forecast. 
\begin{figure}[h]
      \centering
      \caption{Correlation between the simulated market price increments and the renewable production forecast increments in the model during the last  six hours of trading}
      \label{FigPearson}
      \centerline{\includegraphics[width =0.6\textwidth]{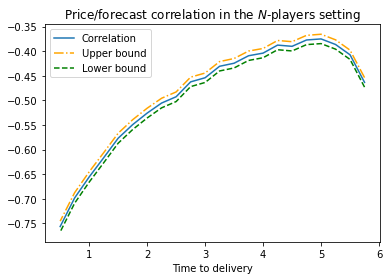}}
\end{figure}

We first note that the correlation is negative: an expected increase of the renewable production is correlated to a decrease in the market price and an expected lack of renewable production is correlated to an increase in the price. As we get closer to the delivery date, the agents trade more actively as new forecast information becomes available, and the market price becomes more strongly dependent on the forecast updates. 
The model outputs qualitatively match the results observed empirically. However, the strength of the correlation seems to be greater in the model than in reality. This can be explained by the fact that the model does not take into account other renewable means of production such as the solar energy. The slight increase of the correlation for longer times to delivery (the right-hand side of the graph) may be explained by the fact that the correlation is computed as the ratio of the covariance to the square root of the product of variances. While both quantities decrease for longer times to delivery, the denominator may decrease faster, explaining the slight increase in the correlation values. 

\paragraph{Convergence and approximations}

In Figure \ref{Fig4} we plot the mean field position, the aggregate $N$-player Nash equilibrium position and the aggregate position for the $\epsilon$-Nash equilibrium (respectively given by Theorem \ref{theorem_complete}, Theorem \ref{theorem_mfg} and Proposition \ref{propepsnashMF}) for a model with $N=5$ players and $N=100$ players. The trajectories were computed with the same simulated fundamental price, common production forecast and parameters as the Figure \ref{Fig2} above, over the 6 hours preceding the delivery time. The left graph ($N=5$) shows a big difference between the Nash equilibrium and $\epsilon$-Nash approximation on one hand, and the mean field on the other hand. This is explained by the individual production forecast taken into account in the Nash and $\epsilon$-Nash equilibria. %It is worth noticing that incomplete information associated with the $\epsilon$-Nash setting does not seem to pull this approximation too far away from the complete information Nash equilibrium even if the aggregate position is replaced by the mean field. 
When we consider a larger number of players, $N=100$, the three position trajectories are much closer to each other. This confirms the asymptotic convergence to the mean field discussed in Section \ref{CV} for the $N$-player Nash equilibrium and $\epsilon$-Nash equilibrium.  
\begin{figure}[h]
      \centering
      \caption{Aggregate position in different settings with $N=5$ (left) and $N=100$ (right) agents}
      \label{Fig4}
      \includegraphics[width =0.49\textwidth]{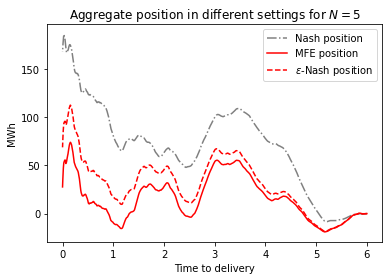}
   %\hfill
   \includegraphics[width = 0.49\textwidth]{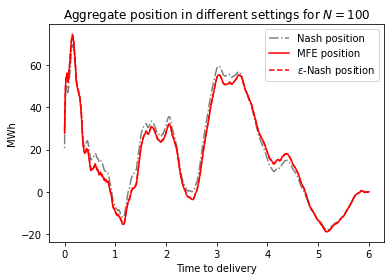}
\end{figure}

\section{Appendix}
\subsection{Proof of Theorem \ref{theorem_complete}}
\label{prooft1}
\begin{proof}
\emph{Step 1. First order condition of optimality for a single agent.} In this step, we are going to show that for fixed $\phi^{-i*}$, the strategy $\phi^{i*}$ satisfies \eqref{nash} if and only if $\mathbb E[\phi^{i*}_T-X^i_T|\mathcal F_0]=0 $ and there exists a square integrable $\mathbb F$-martingale $Y^i$ such that, almost surely,
\begin{align}
\notag & Y^i_t + \alpha(t){\left\{\dot \phi^{i*}_t\left(1-\frac{b_N}{N}\right) +  b_N\dot{\bar\phi}^{N*}_t\right\} } \\ &\qquad + S_t + a(\bar \phi^{N*}_t-\bar \phi^{N*}_0) - \frac{a}{N}(\phi^{i*}_t-\phi^{i*}_0) = 0,\ 0\leq t\leq T, \notag\\  &  Y^i_T  = \frac{a}{N} (\phi^{i*}_T-\phi^{i*}_0) + \lambda(\phi^{i*}_T - X^i_T)-S_T,
\label{mart}
\end{align} 
with the shorthand notation { $b_N = \frac{N}{N-1}\frac{b}{2}$.}
Assume that $\phi^{i*}$ satisfies \eqref{nash}. Then, for any adapted square integrable process $(\nu_t)_{0\leq t\leq T}$, and for any $\delta \in \mathcal F_0$, 
{$$
J^{N,i}(\phi^{i*} +\delta+ \int_0^\cdot \nu_s ds,\phi^{-i*}) \leq J^{N,i}(\phi^{i*},
\phi^{-i*}).
$$}
Developing the expressions, this is equivalent to 
\begin{align*}
&\mathbb E\Big[\int_0^T \nu_t \Big\{\alpha(t){\Big[\dot \phi^{i*}_t\Big(1-\frac{b_N}{N}\Big) +  b_N\dot{\bar\phi}^{N*}_t\Big] }+ S_t + a(\bar \phi^{N*}_t-\bar\phi^{N*}_0) {- \frac{a}{N}(\phi^{i*}_t-\phi^{i*}_0)}\Big\}dt 
\\& + \left(\frac{a}{N} (\phi^{i*}_T-\phi^{i*}_0) + \lambda(\phi^{i*}_T - X^i_T)-S_T\right)\int_0^T \nu_t dt  +\lambda(\phi^{i*}_T - X^i_T)\delta \Big]
\\ &+\mathbb E\Big[\frac{1}{2}\int_0^T \alpha(t)\nu_t^2 dt + \Big(\frac{a}{N} +   \frac{\lambda}{2}\Big)\Big(\int_0^T \nu_t dt\Big)^2 + \frac{\lambda}{2}\delta^2 \Big] \geq 0,\notag
\end{align*}
and since $\nu$ and $\delta$ are arbitrary, we see that optimality is equivalent to 
\begin{align}\notag &\mathbb E\Big[\int_0^T \nu_t \Big\{\alpha(t){\Big[\dot \phi^{i*}_t\Big(1-\frac{b_N}{N}\Big) +  b_N\dot{\bar\phi}^{N*}_t\Big]} + S_t + a(\bar \phi^{N*}_t-\bar\phi^{N*}_0) {- \frac{a}{N}(\phi^{i*}_t-\phi^{i*}_0)}\Big\}dt \\&  + \Big(\frac{a}{N} (\phi^{i*}_T-\phi^{i*}_0)+ \lambda(\phi^{i*}_T - X^i_T)-S_T\Big)\int_0^T \nu_t dt  \Big]=0,\label{optnu}
\end{align}
for any adapted square integrable $\nu$, together with the condition that 
$\mathbb E[\phi^{i*}_T - X^i_T|\mathcal F_0] = 0$. Now, assume that $Y^i$ is a square integrable martingale satisfying \eqref{mart}. Then, by integration by parts, the expression in the previous line equals
\begin{align*}
\mathbb E\Big[-\int_0^T \nu_t Y_t dt + Y_T \int_0^T \nu_t dt\Big] = \mathbb E\Big[\int_0^T \Big(\int_0^t \nu_s ds\Big)dY_t\Big] = 0,
\end{align*}
and we see that the optimality condition is satisfied. Conversely, assume that \eqref{optnu} is satisfied for any adapted square integrable process $\nu$, and let $Y^i$ be a martingale such that 
$$
Y^i_T = \frac{a}{N} (\phi^{i*}_T-\phi^{i*}_0) + \lambda(\phi^{i*}_T - X^i_T)-S_T.
$$
Then, by integration by parts, \eqref{optnu} is equivalent to 
\begin{multline*}
\mathbb E\Big[\int_0^T \nu_t \Big\{\alpha(t){\Big[\dot \phi^{i*}_t\Big(1-\frac{b_N}{N}\Big) +  b_N\dot{\bar\phi}^{N*}_t\big] } \\+ S_t + a(\bar \phi^{N*}_t-\bar\phi^{N*}_0)   {- \frac{a}{N}(\phi^{i*}_t}-\phi^{i*}_0)+ Y^i_t\Big\}dt\Big] = 0,
\end{multline*}
and since $\nu$ is arbitrary, we see that \eqref{mart} is satisfied. 

\emph{Step 2. Computing the average position.} Let $(\phi^{i*})_{i = 1, \dots, N}$ be a Nash equilibrium. We have seen that this is equivalent to \eqref{mart} together with the condition that 
$\mathbb E[\phi^{i*}_T - X^i_T|\mathcal F_0] = 0$ for $i=1,\dots,N$. Summing up these expressions for $i=1,\dots,N$ and denoting $\overline Y^N_t = \frac{1}{N}\sum_{i=1}^N Y^i_t$, we get
\begin{align*}
    &\overline Y^N_t + \alpha(t){\Big(1+\frac{b}{2}\Big)}\dot{\bar\phi}^{N*}_t + S_t + a{\frac{N-1}{N}}(\bar \phi^{N*}_t-\bar\phi^{N*}_0) = 0,\quad 0\leq t\leq T,\\
    &\overline Y^N_T  = \frac{a}{N} (\bar\phi^{N*}_T-\bar \phi^{N*}_0) + \lambda(\bar\phi^{N*}_T - \overline X^N_T)-S_T,\\
    &\mathbb E[\bar\phi^{N*}_T - \overline X^N_T|\mathcal F_0]  = 0.
\end{align*}

The first equation can be solved explicitly for $\bar \phi^{N*}$:
\begin{equation}
   \bar \phi^{N*}_t =\bar \phi^{N*}_0 -\int_0^t \eta^{N}_{s,t}
\frac{\overline Y^N_s + S_s}{\alpha(s){\left(1+\frac{b}{2}\right)}}ds.\label{Agg_position_Y}
\end{equation}
Denoting $\hat\phi_t: = \bar \phi^{N*}_t + I^N_t$,
we obtain simplified equations:
\begin{align*}
\hat \phi_t &= \bar \phi^{N*}_0-\int_0^t  \eta^{N}_{s,t}
\frac{\overline Y^N_s}{\alpha(s){\left(1+\frac{b}{2}\right)}}ds,\\\overline Y^N_T  &= (\frac{a}{N} +\lambda)(\hat \phi_T-I^N_T) -\frac{a}{N}\bar \phi_0^{N*} - \lambda \overline X^N_T-S_T.
\end{align*}

Substituting $\hat \phi_T$ into the second equation and taking the expectation, we obtain another linear equation, this time for $\overline Y^N_t$:

\begin{align*}
\overline Y^N_T & = -(\frac{a}{N} +\lambda) \int_0^T  \eta^{N}_{s,T}
\frac{\overline Y^N_s}{\alpha(s){\left(1+\frac{b}{2}\right)}}ds -(\frac{a}{N} +\lambda)I_T - \lambda \overline X^N_T + \lambda \bar \phi^{N*}_0-S_T,\\
   \overline Y^N_t &= -\left(\frac{a}{N}+\lambda\right)\int_0^t \overline Y^N_s\frac{ \eta^{N}_{s,T}}{\alpha(s){\left(1+\frac{b}{2}\right)}} ds  - \left(\frac{a}{N}+\lambda\right)\Delta^{N}_{t,T}\overline Y^N_t \\& - \left(\frac{a}{N}+\lambda\right) \widetilde I^N_t- \lambda \overline X^N_t+ \lambda \bar \phi^{N*}_0 - \widetilde S_t. 
\end{align*}

By integration by parts, this is equivalent to
\begin{align*}
   \overline Y^N_t = & -\left(\frac{a}{N}+\lambda\right)\int_0^t \Delta^{N}_{s,T}d\overline Y^N_s   - \left(\frac{a}{N}+\lambda\right)\Delta^{N}_{0,T}\overline Y^N_0  \\ &- \left(\frac{a}{N}+\lambda\right) \widetilde I_t- \lambda \overline X^N_t+ \lambda \bar \phi^{N*}_0 - \widetilde S_t.
\end{align*}

Taking $t=0$, we get:
$$
\overline Y^N_0=\frac{ - \left(\frac{a}{N}+\lambda\right) \widetilde I^N_0 - \lambda (\overline X^N_0- \bar \phi^{N*}_0)-\widetilde S_0}{1+\left(\frac{a}{N}+\lambda\right) \Delta^{N}_{0,T}}
$$
On the other hand, in differential form,
$$
\left\{1+\left(\frac{a}{N}+\lambda\right)  \Delta^{N}_{t,T}\right\}d\overline Y^N_t = - \left(\frac{a}{N}+\lambda\right)  d\widetilde I^N_t - \lambda d\overline X^N_t-d\widetilde S_t,
$$

which is solved therefore explicitly by
\begin{equation}\label{Ycomplete}
    \overline Y^N_t  = \frac{ - \left(\frac{a}{N}+\lambda\right) \widetilde I^N_0 - \lambda (\overline X^N_0- \bar \phi^{N*}_0)-\widetilde S_0}{1+\left(\frac{a}{N}+\lambda\right) \Delta^{N}_{0,T}}- \int_0^t \frac{ \left(\frac{a}{N}+\lambda\right)  d\widetilde I^N_s+ \lambda d\overline X^N_s+d\widetilde S_t}{1+\left(\frac{a}{N}+\lambda\right)  \Delta^{N}_{s,T}}
\end{equation}

%and also
%$$
%\overline Y^N_t  +S_t = S_t - \widetilde S_t + \frac{ \widetilde S_0 - \lambda \overline X^N_0}{1+\left(\frac{a}{N}+\lambda\right) \Delta_{0,T}} + \int_0^t \frac{d\widetilde S_s - \left(\frac{a}{N}+\lambda\right)  (\widetilde S_s - S_s)d\Delta_{s,T}- \lambda d\overline X^N_s}{1+\left(\frac{a}{N}+\lambda\right)  \Delta_{s,T}}.
%$$
Finally
\begin{equation}
    \begin{aligned}
    \bar \phi^{N*}_t &=\bar \phi^{N*}_0 -I^N_t + \int_0^t \overline Y^N_s d\Delta^{N}_{s,t} =\bar \phi^{N*}_0-I^N_t - \overline Y^N_0  \Delta^{N}_{0,t} - \int_0^t \Delta^{N}_{s,t} d \overline Y^N_s\\
    & = \bar \phi^{N*}_0- I^N_t + \Delta^{N}_{0,t}\frac{ \left(\frac{a}{N}+\lambda\right) \widetilde I^N_0 + \lambda (\overline X^N_0-\bar \phi^{N*}_0)+\widetilde S_0}{1+\left(\frac{a}{N}+\lambda\right) \Delta^{N}_{0,T}} \\
    &+ \int_0^t \Delta^{N}_{s,t}\frac{ \left(\frac{a}{N}+\lambda\right)  d\widetilde I^N_s+ \lambda d\overline X^N_s+ d\widetilde S_t}{1+\left(\frac{a}{N}+\lambda\right)  \Delta^{N}_{s,T}}.\label{Agg_position}
    \end{aligned}
\end{equation}
It remains to compute $\bar \phi^{N*}_0$ from the condition $\mathbb E[\bar\phi^{N*}_T - \overline X^N_T|\mathcal F_0] = 0$. Substituting this into the above expression, we find
$$
\bar \phi^{N*}_0=\overline X^N_0+\frac{\widetilde I^N_0 - \Delta^{N}_{0,T} \widetilde S_0}{1+\frac{a}{N} \Delta^{N}_{0,T}}
$$
so that
\begin{align}
     \overline\phi^{N*}_t   & = \overline X^N_0+\frac{1+\frac{a}{N} \Delta^{N}_{0,t}}{1+\frac{a}{N} \Delta^{N}_{0,T}}(\widetilde I^N_0 - \Delta^{N}_{0,T} \widetilde S_0) \notag\\
    &- (I^N_t-\Delta^N_{0,t}\widetilde S_0)+ \int_0^t \Delta^{N}_{s,t}\frac{ \left(\frac{a}{N}+\lambda\right)  d\widetilde I^N_s+ \lambda d\overline X^N_s+ d\widetilde S_t}{1+\left(\frac{a}{N}+\lambda\right)  \Delta^{N}_{s,T}}.\label{avstrat}
\end{align}    

\emph{Step 3: computing the position of the agent.} 
Let $\check \phi^{i*}_t:= \phi^{i*}_t - \bar \phi^{N*}_t$, \\ $\check X^i_t = X^i_t - \overline X^N_t$ and $\check Y^i_t:= Y^i_t -
\overline Y^N_t$. Then, $\check Y^i$ is an $\mathbb F$-martingale
and satisfies
$$
\check Y^i_T = \frac{a}{N}\check\phi^{i*}_T + \lambda(\check \phi^{i*}_T -\check X^i_T),\qquad
\check Y^i_t = -\alpha(t) \dot{\check \phi}^{i*}_t {+ \frac{a}{N}\check \phi^{i*}_t.}
$$
together with the additional condition $\mathbb E[\check \phi^{i*}_T - \check X^i_T|\mathcal F_t] = 0$. 
Similarly to the second part, this system admits an explicit solution:
$$
\check Y^i_t = -\frac{\lambda (\check X^i_0-\check \phi^{i*}_0)}{1+
  \left(\frac{a}{N}+\lambda\right)\widetilde\Delta^{N}_{0,T}} -\int_0^t\frac{\lambda d\check X^i_s}{1+
  \left(\frac{a}{N}+\lambda\right)\widetilde\Delta^{N}_{s,T}}.
$$
and 
$$
\check \phi^{i*}_t = \check X^i_0 + \int_0^t \widetilde \Delta^{N}_{s,t} \frac{\lambda d\check X^i_s}{1+
  \left(\frac{a}{N}+\lambda\right)\widetilde\Delta^{N}_{s,T}}  . 
$$

\end{proof} 

\begin{remark}
The optimal strategy is obtained by solving a series of linear equations from an equivalent characterization of optimality. Since all equations admit unique solutions, the equilibrium strategy is unique. This is a consequence of the strict concavity of the objective function in our linear quadratic setting. 
\end{remark}

\subsection{Proofs of Propositions \ref{CVNash} and \ref{propepsnashMF}
}\label{ProofCVNash}

\paragraph{Proof of Proposition \ref{CVNash}.}

%Let us first prove that at the equilibrium, at each time $t$, the average position $\frac{1}{N}\Phi_t$ in the complete information setting, converges pointwise in $\mathbb{L}^2$ to the mean field $\bar \phi_t$.

For all $ t \in [0,T]$, we define: 
%Without loss of generality assume $C=1$ define: 
\begin{align*}
    & g^N_{s,t} =\frac{\Delta^{N}_{s,t}}{(1+(\frac{a}{N}+\lambda)\Delta^{N}_{s,T})}, \quad g_{s,t} =\frac{\Delta_{s,t}}{(1+\lambda\Delta_{s,T})}, \\
    & \tilde g^{N}_{s,t} =\frac{\widetilde\Delta^{N}_{s,t}}{(1+(\frac{a}{N}+\lambda)\widetilde\Delta^N_{s,T})}, \quad \tilde g_{s,t} =\frac{\widetilde \Delta_{s,t}}{(1+\lambda\widetilde \Delta_{s,T})},
\end{align*}
so that, for some constant $C$ depending only on the parameters $a$, {$b$}, $\alpha$ and $\lambda$, but not on other ingredients of the model,
$$
|\Delta^N_{s,t}-\Delta_{s,t}|+|g^N_{s,t}-g_{s,t}| + |\tilde g^N_{s,t} - \tilde g_{s,t}| \leq \frac{C}{N}. 
$$
for all $s,t\in[0,T]$. 
Now, let us consider the optimal strategies $(\phi_t^{i*})_{t\in [0,T]}$ and  $(\phi_t^{MF,i*})_{t\in [0,T]}$ of the generic agent $i$ respectively in the $N$-player setting and the mean field setting. Fix $t \in [0,T]$. Then, 
\begin{align*}
&\phi_t^{i*}-\phi_t^{MF,i*} =  \frac{a}{N}\frac{\Delta^N_{0,t}-\Delta^N_{0,T}}{1+\frac{a}{N}\Delta^N_{0,T}}(\widetilde I^N_0 - \widetilde S_0 \Delta^N_{0,T}) + \widetilde I^N_0 - \widetilde I_0 + \widetilde S_0(\Delta_{0,T}-\Delta^N_{0,T})\\ & - I^N_t+ I_t + \widetilde S_0(\Delta^N_{0,t}-\Delta_{0,t})+ \int_0^t (g^N_{s,t}-g_{s,t})\left\{ \left(\frac{a}{N}+\lambda\right)  d\widetilde I^N_s+ \lambda d\overline X^N_s + d\widetilde S_s\right\}\\ &+\lambda\int_0^t g_t(s) d( \widetilde I^N_s-\widetilde I_s+ \overline X^N_s-\overline X_s)\\ &+\int_0^t \lambda(\tilde g^N_{s,t}-\tilde g_{s,t}) d(X^i_s - \overline X^N_s) + \lambda\int_0^t \tilde g_{s,t}  d( \overline X_s- \overline X^N_s) 
\end{align*}
Therefore, for some constant $C$ depending only on the parameters $a$, {$b$}, $\alpha$ and $\lambda$, but not on other ingredients of the model,
\begin{align*}
\mathbb E[(\phi_t^{i*}-\phi_t^{MF,i*} )^2]&\leq \frac{C}{N^2} \mathbb E[(\widetilde I^N_0)^2] + \frac{C}{N^2} \mathbb E[(\widetilde S_0)^2]  + \mathbb E[(\widetilde I^N_0 - \widetilde I_0)^2] + \mathbb E[(I^N_t - I_t)^2] \\ &+ \frac{C}{N^2} \mathbb E[(\widetilde I^N_t)^2] + C\mathbb E[(\widetilde I^N_t - \widetilde I_t)^2]  +\frac{C}{N^2}\mathbb E[\widetilde S_t^2]\\ &+ \frac{C}{N^2} \mathbb E[(\overline X_t^N)^2] + C\mathbb E[(\overline X^N_t - \overline X_t)^2] + \frac{C}{N^2} \mathbb E[(\check X^i_t)^2]\\
&\leq \mathbb E[(I^N_t - I_t)^2] + \frac{C}{N^2} \mathbb E[( I^N_T)^2] + C\mathbb E[(I^N_T -  I_T)^2]\\ &+\frac{C}{N^2} \mathbb E[S_T^2]+\frac{C}{N^2} \mathbb E[(\overline X_t)^2] + \frac{C}{N^2} \sum_{i=1}^N\mathbb E[(\check X_t^i)^2]\\
&\leq \frac{C}{N^2}  \mathbb E[\sup_{0\leq s\leq T} S_s^2]  +  \frac{C}{N^2} \mathbb E[(\overline X_t)^2] + \frac{C}{N} \mathbb E[(\check X_t^i)^2].
\end{align*}
where the estimate for the first line above is obtained through Jensen's inequality. The other estimates of the proposition are obtained in a similar way. 

\paragraph{Proof of Proposition \ref{propepsnashMF}.} To lighten notation, and since we now have only one strategy, we omit in this proof the superscript MF in the candidate strategy $\phi^{MF,i*}$. The "distance to optimality" for this strategy is estimated as follows. 
\begin{align*}
&J^{N,i}(\phi^{i},\phi^{-i*}) - J^{N,i}(\phi^{i*},\phi^{-i*}) \\
&= J^{N,i}(\phi^{i},\phi^{-i*}) - J^{MF}(\phi^{i},\bar\phi^*)+J^{MF}(\phi^{i},\bar\phi^*) - J^{MF}(\phi^{i*},\bar\phi^*)\\&\quad + J^{MF}(\phi^{i^*},\bar\phi^*) - J^{N,i}(\phi^{i*},\phi^{-i*})\\
&\leq J^{N,i}(\phi^{i},\phi^{-i*}) - J^{MF}(\phi^{i},\bar\phi^*) + J^{MF}(\phi^{i^*},\bar\phi^*) - J^{N,i}(\phi^{i*},\phi^{-i*}).
\end{align*}
The second difference is estimated as follows:
\begin{align}
&J^{MF}(\phi^{i^*},\bar\phi^*) - J^{N,i}(\phi^{i*},\phi^{-i*}) \notag\\
= & -\mathbb E\Big[\int_0^T \dot\phi^{i*}_t\Big\{ a\big(\bar \phi^*_t - \bar \phi^{N*}_t-\bar \phi^*_0 + \bar \phi^{N*}_0\big) {+ \frac{\alpha(t)b}{2}\big(\dot{\bar \phi}^*_t-\dot{\bar \phi}^{N,-i*}_t\big)}\Big\}dt\Big]
\end{align}
Since 
\begin{align}
\dot{\bar \phi}^*_t-\dot{\bar \phi}^{N,-i*}_t = \dot{\bar \phi}^*_t-\dot{\bar \phi}^{N*}_t + \frac{\dot\phi^{i*}-\dot{\bar \phi}^{N*}}{N-1},\label{addest}
\end{align}
it follows from the Cauchy-Schwarz inequality and Proposition \ref{CVNash} that 
\begin{align}
|J^{MF}(\phi^{i^*},\bar\phi^*) - J^{N,i}(\phi^{i*},\phi^{-i*}) |\leq \frac{C}{\sqrt{N}}.\label{cauchyschwarz}
\end{align}

%\begin{align}
%\leq & \mathbb E\left[\int_0^T (\dot\phi^{i*}_t)^2 dt\right]^{\frac{1}{2}}\left\{a \mathbb E\left[\int_0^T (\bar \phi^*_t - \bar \phi^{N*}_t)^2 dt\right]^{\frac{1}{2}}\right.\notag\\
%&\left.{+ b \mathbb E\left[\int_0^T \alpha^2(t)\left(\dot{\bar \phi}^*_t-\dot{\bar \phi}^{N}_t\right)^2 dt\right]^{\frac{1}{2}} +\frac{b_N}{N} \mathbb E\left[\int_0^T\alpha^2(t) (\dot\phi^{i*}_t)^2 dt\right]^{\frac{1}{2}}} \right\},\label{cauchyschwartz}
%\end{align}
%\todo{Je n'obtiens pas la même estimation dans (23), le dernier terme est différent, peux-tu vérifier?}
%where\todo{pourquoi la borne (26) n'est pas celle de la proposition 1?}
%\begin{align}
%\mathbb E\left[\int_0^T (\bar \phi^*_t - \bar \phi^{N*}_t)^2dt \right]  \leq  \frac{1}{N^2}\sum_{i=1}^N \mathbb E[(\check X^i_T)^2] = \frac{1}{N} \mathbb E[(\check X^1_T)^2]{\le \frac{\widetilde C}{N}}, \label{estphibar}
%\end{align}
%{and from Proposition \ref{CVNash},
% \begin{align}
%\mathbb E\left[\int_0^T \alpha^2(t)(\dot{\bar \phi}^*_t - \dot{\bar \phi}^{N*}_t)^2dt \right] & \leq   \frac{C}{N^2} \int_0^T \mathbb E[S_t^2] dt +  \frac{C}{N^2} \mathbb E[(\overline X_T)^2] + \frac{C}{N} \mathbb E[(\check X_T^i)^2]\notag\\
%&\le \frac{\widetilde C}{N}, \label{eq:traderatephibar}
%\end{align}
%where $\widetilde C$ is a constant that does not depend on $N$.}

The first difference admits the following estimate.
\begin{align}
&J^{N,i}(\phi^{i},\phi^{-i*}) - J^{MF}(\phi^{i},\bar\phi^*) 
= a\mathbb E\Big[\int_0^T \dot\phi^{i}_t (\bar \phi^*_t - \bar \phi^{N*}_t-\bar \phi^*_0 + \bar \phi^{N*}_0)dt\Big]\notag\\ &- \frac{a}{N}\mathbb E\Big[\int_0^T \dot\phi^{i}_t (\phi^i_t-\phi^{i*}_t-\phi^i_0+\phi^{i*}_0)dt\Big]+ b\mathbb E \Big[\int_0^T\alpha(t)\dot\phi^{i}_t (\dot{\bar\phi}^*_t-\dot{\bar\phi}^{N,-i*}_t)dt\Big] \notag\\
&\leq \mathbb E\Big[\int_0^T (\dot\phi^{i}_t)^2 dt\Big]^{\frac{1}{2}}\Big\{a \mathbb  E\Big[\int_0^T (\bar \phi^*_t - \bar \phi^{N*}_t-\bar \phi^*_0 + \bar \phi^{N*}_0)^2 dt\Big]^{\frac{1}{2}}  \notag\\
& +  \frac{a}{N} \mathbb E\Big[\int_0^T (\phi^{i*}_t-\phi^{i*}_0)^2 dt\Big]^{\frac{1}{2}}+ \frac{b}{N-1} \mathbb E\Big[\int_0^T \alpha^2(t)(\dot\phi^{i*}_t-\dot{\bar\phi}^{N*})^2 dt\Big]^{\frac{1}{2}} \notag\\&+ b \mathbb E\Big[ \int_0^T\alpha^2(t)(\dot{\bar \phi}^*_t-\dot{\bar \phi}^{N,*}_t)^2 \Big]^\frac{1}{2}
\Big\}\leq  \frac{C_0}{\sqrt{N}}\mathbb E\Big[\int_0^T (\dot\phi^{i}_t)^2 dt\Big]^{\frac{1}{2}} \label{smallint}
\end{align}
for some constant $C_0<\infty$, which does not depend on $\phi^i$, in view of Proposition \ref{CVNash} and \eqref{addest}. 
On the other hand, the following estimate also holds true. 
\begin{align}
           &J^{N,i}(\phi^i, \phi^{-i*}) = - \mathbb{E}\Big[ \int_{0}^{T}\Big\{\frac{\alpha(t)}{2}\Dot{\phi^i_t}\Big(\dot\phi^i_t + b \dot{\bar \phi}^{N,-i*}_t\Big)+\Dot{\phi^i_{t}}(S_t + a(\bar \phi^{N*}_t - \bar\phi^{N*}_0) \Big\}dt \notag\\ & +\frac{a}{N}\int_0^T\dot\phi^i_t(\phi^i_t - \phi^i_0 - \phi^{i*}_t +  \phi^{i*}_0)dt +\phi^i_0 S_0 - (\phi^i_T-X^i_T)S_T+\frac{\lambda}{2}(\phi^i_{T}- X^i_{T})^2\Big],\notag\\
           &\leq -\frac{\bar\alpha}{2}\mathbb E\Big[\int_0^T (\dot\phi^i_t)^2dt\Big]  - \frac{\lambda}{2}\mathbb E[(\phi^i_T)^2]+ C_1 \mathbb E\Big[\int_0^T (\dot\phi^i_t)^2dt\Big]^{1\over 2} + C_2 \mathbb E[(\phi^i_T)^2]^{1\over2} + C_3,
\end{align}
where $\bar{\bar \alpha} = \max_{0\leq t\leq T}\alpha(t)$ and 
\begin{align*}
C_1 &= \frac{b}{2} \mathbb E\Big[\int_0^T\alpha^2(t) (\dot{\bar\phi}^{N,-i*}_t)^2dt\Big]^{1\over 2}  + a \mathbb E\Big[\int_0^T (\bar \phi^{N*}_t-\bar\phi^{N*}_0)^2dt\Big]^{1\over 2}\\  &+ \frac{a}{N}\mathbb E\Big[\int_0^T (\bar \phi^{i*}_t-\bar\phi^{i*}_0)^2dt\Big]^{1\over 2} + \mathbb E\Big[\int_0^T (S_t-S_0)^2\Big]^{1\over 2},\\
C_2 & = \mathbb E[(S_T-S_0)^2]^{1\over2} + 2 \mathbb E[(X^i_T)^2]^{\frac{1}{2}},\qquad 
C_3 = |\mathbb E[X^i_T S_T]|. 
\end{align*}
Thus, there exists a constant $C^*<\infty$, which does not depend on $\phi^i$, such that if 
$$
\mathbb E\Big[\int_0^T (\dot\phi^i_t)^2dt\Big]  +\mathbb E[(\phi^i_T)^2]>C^*, 
$$
then 
$$
J^{N,i}(\phi^i, \phi^{-i*}) - J^{N,i}(\phi^{i*}, \phi^{-i*})<0. 
$$
Therefore, from \eqref{cauchyschwarz} and \eqref{smallint} it follows that for any admissible strategy $\phi^i$, 
\begin{align*}
&J^{N,i}(\phi^i, \phi^{-i*}) - J^{N,i}(\phi^{i*}, \phi^{-i*}) \\ &\leq \mathbf 1_{\mathbb E[\int_0^T (\dot\phi^i_t)^2dt ] +\mathbb E[(\phi^i_T)^2]\leq C^*}\Big\{\frac{C}{\sqrt{N}} + \frac{C_0}{\sqrt{N}} \mathbb E[\int_0^T(\dot\phi_t^i)^2 dt]^{1\over 2}\Big\}\\
&\leq \frac{C}{\sqrt{N}} + \frac{C_0 \sqrt{C^*}}{\sqrt{N}}. 
\end{align*}
\end{document}